\documentclass[10pt, journal, letterpaper]{IEEEtran}

\usepackage[utf8]{inputenc}
\usepackage{graphicx}
\usepackage{listings}
\usepackage{amsmath}
\usepackage{amssymb}
\usepackage{stackrel}
\usepackage{arydshln}
\usepackage{framed}
\ifCLASSOPTIONcompsoc
    \usepackage[caption=false, font=normalsize, labelfont=sf, textfont=sf]{subfig}
\else
\usepackage[caption=false, font=footnotesize]{subfig}
\fi

\usepackage{algorithm}
\usepackage{cite}

\usepackage{amsthm}
\newtheorem{lemma}{Lemma}

\newtheorem{theorem}{Theorem}

\usepackage{algorithmic}
\usepackage{hyperref}
\usepackage[shortlabels]{enumitem}
\usepackage{bigints}
\usepackage{xcolor}

\DeclareMathOperator*{\argmax}{arg\,max}

\newcommand{\E}{\mathbb{E}}
\newcommand{\B}{\mathcal{B}}

\newcommand{\tr}[1]{\text{tr}\left(#1\right)}

\renewcommand{\b}[1]{\mathbf{#1}}
\newcommand{\R}{\mathbb{R}}
\newcommand{\N}{\mathbb{N}}

\begin{document}

\title{Monitoring Correlated Sources: AoI-based Scheduling is Nearly Optimal}

\author{R Vallabh Ramakanth, Vishrant Tripathi and Eytan Modiano
\thanks{R Vallabh Ramakanth, Vishrant Tripathi and Eytan Modiano are with the Laboratory of Information and Decision Systems (LIDS), Massachusetts Institute of Technology, Cambridge, MA, 02139. A preliminary version of this paper appeared in the conference proceedings of IEEE INFOCOM 2024.\\
Email: \{vallabhr, vishrant, modiano\}@mit.edu.}}

\date{}
\maketitle

\begin{abstract}
We study the design of scheduling policies to minimize monitoring error for a collection of correlated sources, where only one source can be observed at any given time. We model correlated sources as a discrete-time Wiener process, where the increments are multivariate normal random variables, with a general covariance matrix that captures the correlation structure between the sources. Under a Kalman filter-based optimal estimation framework, we show that the performance of all scheduling policies oblivious to instantaneous error, can be lower and upper bounded by the weighted sum of Age of Information (AoI) across the sources for appropriately chosen weights. We use this insight to design scheduling policies that are only a constant factor away from optimality, and make the rather surprising observation that AoI-based scheduling that ignores correlation is sufficient to obtain performance guarantees. We also derive scaling results that show that the optimal error scales roughly as the square of the dimensionality of the system, even in the presence of correlation. Finally, we provide simulation results to verify our claims.

\end{abstract} 
	\section{Introduction}\label{chap:intro}
\subsection{Motivation}\label{chap:intro-sec:motiv}
Remote monitoring and estimation of physical processes have been receiving growing interest in the field of wireless networks. Timely estimation of physical processes and up-to-date knowledge of system state over capacity-limited communication channels is crucial in many applications such as sensing in IoT, control of robot swarms, communication between autonomous vehicles, environmental monitoring, search-and-rescue, and real-time surveillance. Fresh and up-to-date information regarding the system state is necessary for good monitoring and control performance in such tasks.

A common, yet powerful, way to model many physical processes is to use the linear time-invariant (LTI) system framework. A simplified version of the discrete LTI system is the Wiener process in discrete time. We primarily focus on multi-dimensional discrete-time Wiener process in this work. Let us denote the state of this physical process at time $t$ as $\b{x}_t \in \mathbb{R}^{M}$, where $M$ is the number of sources (sensors) being monitored.

Many physical processes of interest exhibit some form of correlation in their dynamics. For example, in the context of environment monitoring, the data collected by sensors measuring the temperature or humidity of a geographical location would be correlated. Hence, we assume that the $M$ dimensional Wiener process evolves in a \textit{correlated} manner. In particular, the increments of the Wiener process are assumed to be multi-variate normal random vectors with a general covariance matrix, that captures the correlation structure between the different sensors.

Many physical processes can be modelled well using the discrete-time linear time-invariant (LTI) system model, where the system state $\b{x}_t \in \R^M$ at time $t$  evolves as 
\begin{align}
    \b{x}_{t+1} = A\b{x}_t + \b{w}_t. \label{eqn:intro_sys_model}
\end{align}
Here, $A$ denotes the system matrix that describes the time evolution of the system and $\b{w}_t$ denotes the random innovations in the process at time $t$, usually modeled as multivariate normal random variables that are i.i.d. across time. 

However, in this work, we primarily look at a simplified version of this model where the system matrix $A$ is set to the identity matrix $I$. This means that the processes of our interest evolve as Gaussian random walks
\begin{align}
    \b{x}_{t+1} = \b{x}_t + \b{w}_t. \label{eqn:intro_random_walks}
\end{align}
Most importantly, we assume that $\b{w}_t \in \mathbb{R}^{M}$ is an i.i.d. multi-variate normal random variable across time, i.e. $\b{w}_t \sim \mathcal{N}(0,Q)$. Here $Q$ is the covariance matrix for the noise innovations. Choosing a non-diagonal $Q$ matrix allows us to model correlation between the random walks.

We study a remote monitoring setting, where a monitor located at the base station is interested in tracking the state of this $M$ dimensional correlated Wiener process as accurately as possible; i.e., our goal is to minimize the time-average monitoring error. Due to bandwidth and interference constraints of the wireless channel, we assume that only one dimension of the state $\b{x}_t$ can be observed in any given time-slot. In other words, while there are $M$ available sensors, one for each dimension of the process, the monitor can only observe one at any give time.  

Solving this problem involves two tasks - estimation and scheduling. \textbf{Estimation} involves reconstructing an estimate $\hat{\b{x}}_t$ for the current state of the system based on all previously received status updates at the base station. \textbf{Scheduling} involves choosing the sensor that gets to transmit its state to the base station in each time-slot. Throughout this work, we will focus on the problem of designing centralized scheduling policies, where the base station decides which sensor to schedule at the beginning of each time slot. Thus, the base station makes the scheduling decision without knowledge to the instantaneous error or the current state of the process. We call this the \textit{oblivious} scheduling setting.

At first glance, it might seem that the problem of minimizing monitoring error requires the solution to a complicated joint optimization problem over the space of estimators and scheduling policies. However, we show in Section \ref{chap:kalmanfilter} that within the class of oblivious scheduling policies, the overall problem is \textit{separable}, and estimation and scheduling can be solved independently of one another.
Further, the optimal estimator design problem in this setting is well-studied in control literature. It involves the application of the Kalman filter approach for \textit{time-varying} LTI systems. 
The more interesting problem, and the focus of our work, is the problem of designing an optimal scheduling policy given the optimal estimator.

Intuitively, a scheduling policy that prioritizes the collection of the ``freshest"  information from the sensors should work well for minimizing error. 
One way to characterize the information freshness of sensors is using the well-known Age of Information (AoI) metric \cite{kaul_real-time_2012}. In fact, we show that for correlated Gaussian sources, the monitoring error of any oblivious scheduling policy can be lower and upper bounded by the weighted sum of AoI of all the sensors, with appropriately chosen weights. These bounds hold on a per time slot and sample path basis, establishing a direct relationship between minimizing weighted-sum AoI and monitoring error, similar to the uncorrelated case. Using this insight, we design Max-Weight style policies based on AoI that have provable constant factor performance guarantees in terms of monitoring error performance for infinite time horizons. This is a surprising result since it suggests that at least for scheduling design, one can effectively ignore correlation and still get near-optimal performance.

\subsection{Prior Works}\label{chap:intro-sec:priorworks}
The connection between remote estimation and AoI has been well documented in prior works\cite{yates_age_survey_2021, champati2019performance, klugel2019aoi}. In fact, for a one-dimensional Wiener process, minimization of estimation error is equivalent to minimization of AoI, for policies which are oblivious to instantaneous error\cite{sun_yuri_remote_estimation}. The same equivalence between monitoring error and weighted-sum AoI minimization appears for \textit{uncorrelated} linear Gaussian random processes in \cite{tripathi2019whittle}. In \cite{ornee_sampling_2021, ornee_whittle_2023}, these ideas are extended to the monitoring of single and multiple uncorrelated Ornstein-Uhlenbeck processes respectively.

There has also been work on distributed scheduling algorithms for minimizing monitoring error. In \cite{chen_real-time_2021}, the authors develop an ALOHA-like scheduling strategy for monitoring \textit{uncorrelated} discrete-time Wiener processes in both the oblivious and non-oblivious setting. They show that their proposed policies are asymptotically constant-factor optimal. 

Works on wireless scheduling for timely estimation typically study only independent processes. The performance of monitoring correlated sources has not been explored thoroughly. Intuitively, correlation between sources can potentially allow us to get information about the system state by only measuring a few sources. This intuition motivates us to find scheduling policies which make the best use of bandwidth resources in order to minimize monitoring error by exploiting correlation.
A simple toy example would be the monitoring of a discrete-time Wiener process with 10 sources. Let these sources be numbered from 1 to 10. If we assume pairs of sources $\{(1,2), (3,4), (5,6), (7,8), (9,10)\}$ to be perfectly correlated, then we get perfect system knowledge by only scheduling sensors 1,3,5,7 and 9. However, if the sources were only partially correlated, then intuitively, we would need to schedule each source as we cannot get away with sampling only a subset of these sources. In this work, we formalize this intuition via performance bounds.

Nevertheless, there has been some preliminary work on the remote monitoring of correlated sources. In particular, a simple probabilistic correlation model is considered in \cite{tripathi_optimizing_2022}. This simple model allows the authors to design scheduling policies with performance guarantees and provide scaling results that take correlation into account. In \cite{he_camera_net_2018} and \cite{he_camera_net_2019}, the authors consider a network of cameras with overlapping fields-of-view and formulate an optimization problem based on AoI to study processing and scheduling in this correlated setting.

There has also been some preliminary work on remote monitoring of LTI systems in the control theory literature \cite{mo_infinite-horizon_2014, shi_scheduling_2012,  han_stochastic_2017, mo_convex_optimization_2009, mo_sensor_2011}. Scheduling of sensors for optimal error finite time horizon has been studied in \cite{Finite_Horizon_Vitus_2012}, wherein the authors exploit the structure of the problem to reduce the search space of schedules without loss in optimality. In \cite{mo_infinite-horizon_2014}, the sensor scheduling problem for the the infinite time horizon is shown to be computationally infeasible to solve \textit{exactly} due its combinatorial nature. However, when there are only two sources, the authors of \cite{shi_scheduling_2012} find an explicit periodic scheduling policy which is optimal in the infinite time horizon. In \cite{mo_infinite-horizon_2014, infinite_horizon_Zhao_2014}, the authors also show that periodic schedules can approximate the performance of any other schedule to arbitrary precision for the case of infinite time horizons, and hence the method developed in \cite{Finite_Horizon_Vitus_2012} can be used to produce a schedule arbitrarily close to optimal. Though the results of \cite{Finite_Horizon_Vitus_2012} help in reducing the space of schedules that needs to be searched, it is still computationally expensive to carry it out for a large sensor network.

To the best of our knowledge, this is the first work to establish the connection between AoI optimization and monitoring error for the case of \textit{correlated} sources - leading to insightful and low complexity scheduling policies with error performance guarantees. 

\subsection{Outline}
The rest of the paper is organized as follows. In Section \ref{chap:problem_formulation}, we describe the system model and explicitly formulate the problem. In Section \ref{chap:kalmanfilter}, we discuss the optimality of the Kalman filter estimator when the scheduler is confined to oblivious policies. In Section \ref{chap:aoi-error-relation}, we develop upper and lower bounds for the estimation error of any oblivious scheduling policy in terms of a weighted sum of AoIs across all sensors. Using these bounds, we develop Max-Weight style policies in Section \ref{chap:schedule_policy} and provide constant factor performance guarantees. We further use these bounds to show that correlation does not lead to any order improvement in performance when compared to the non-correlated case in Section \ref{chap:scaling}. Finally, in Section \ref{chap:simulations}, we provide simulation results that support our theoretical results.

\subsection{Notation}
Vectors are denoted in boldface, whereas scalars and matrices are denoted in plainface, and the difference between matrices and scalars is usually made explicitly. If $\b{x}$ is a $d$ dimensional vector ($\b{x} \in \R^d$), then $x^i$ denotes the $i^{th}$ coordinate of the vector. Usually, the subscript $t$ or $k$ denotes the time index (for example, $x^i_t$, $\b{x}_t$). If $A$ is a matrix, then $A^T$ denotes the transpose of matrix $A$. $\tr{A}$ denotes the trace of matrix $A$, and if $A$ is invertible, $A^{-1}$ denotes the matrix inverse of $A$. If $A \in \R^{M\times M}$ is a real symmetric matrix, then, $\lambda_i(A)$ denotes the $i^{th}$ largest eigenvalue of $A$ (since all the eignevalues of real symmetric matrices are real). More specifically, $\lambda_M(A) \leq \lambda_{M-1}(A) \leq ... \leq \lambda_1(A)$. $A \succeq 0$ denotes that $A$ is a positive semidefinite matrix, and $A \succ 0$ denotes that $A$ is a positive matrix. If $A \succeq B$, then $A-B \succeq 0$. If $\b{x} \sim \N(\b{\mu}, \Sigma)$, where $\mu \in \R^M$ and $Q \in \R^{M\times M}, Q \succeq 0$, then $\b{x}$ is a multivariate normal random variable with mean $\b{\mu}$ and covariance $\Sigma$.
 
        \section{Problem Formulation}\label{chap:problem_formulation}
\subsection{System Model}\label{chap:problem_formulation-sec:system_model}
Consider $M$ sensors communicating over a wireless channel to a central base station (Fig. \ref{fig:sys_model}). The base station is interested in keeping track of the state of the process at each sensor, with as little monitoring error as possible. Due to interference and bandwidth constraints, we assume that only one of these $M$ sensors can communicate with the base station at any given time.
\begin{figure}[!hbt]
    \centering
    \includegraphics[scale=0.35]{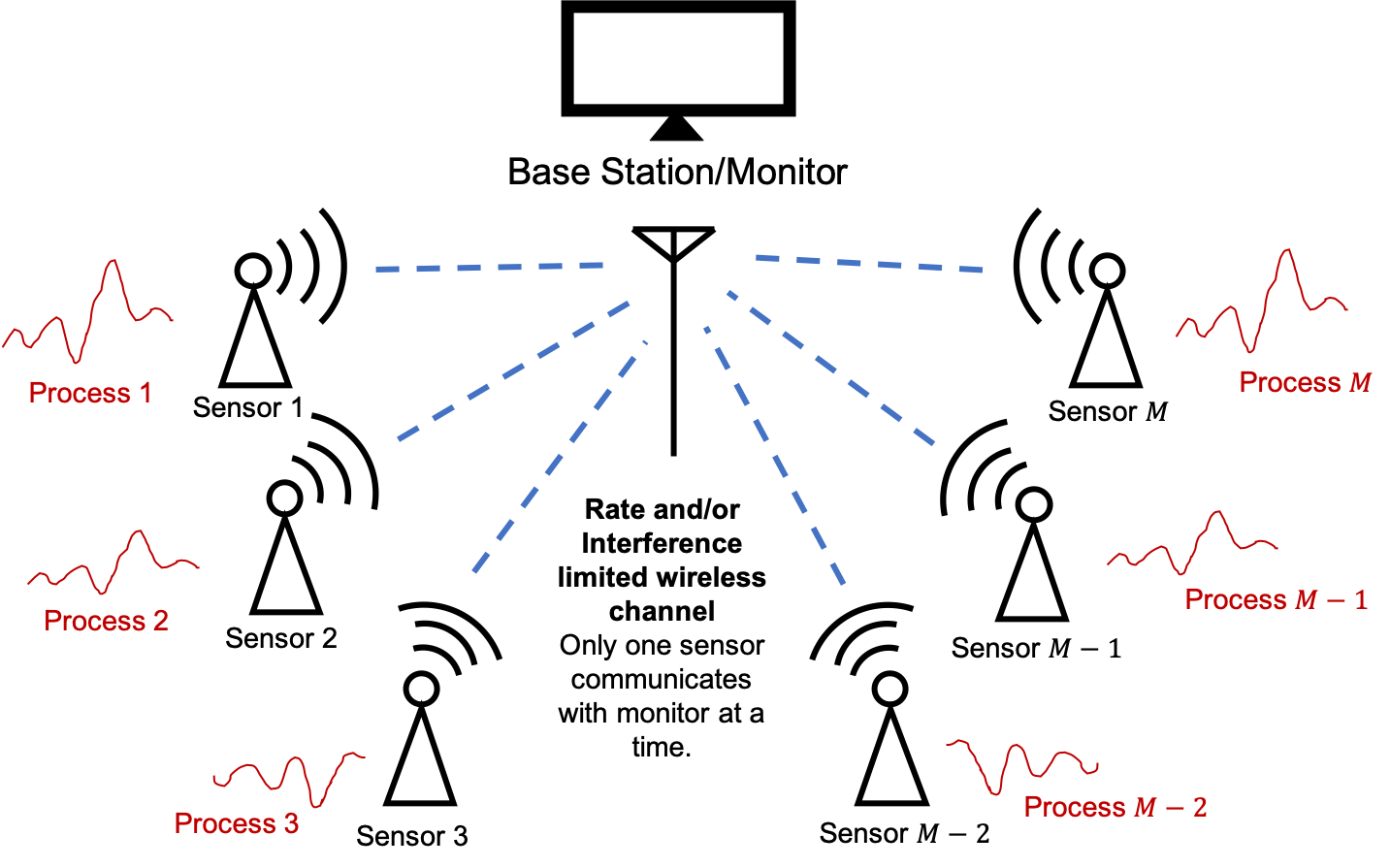}
    \caption{$M$ sensors tracking $M$ processes which are correlated. The central monitor can communicate with sensors one at a time, due to interference constraints.}
    \label{fig:sys_model}
\end{figure}

We assume that the dynamics of the process being monitored at each sensor are linear and time-invariant discrete time random walks with Gaussian increments (i.e., discrete-time Wiener process). More, specifically, the state of the $i$th process evolves as follows
\begin{align}
    x^i_{t+1} = x^i_{t} + w^i_{t},
\end{align}
where, $ i\in\{1,...,M\}$ and $w^i_t$ is the zero mean Gaussian increment at time $t$ for the $i$th process.

The most important aspect in our model is the assumption that the increments across different sources are \textit{correlated}. To model this correlation, we assume that the vector of increments $\b{w}_t$ is a zero-mean multivariate normal random variable with a covariance matrix that is possibly non-diagonal. If we aggregate the state evolution equations for all sources together in vector form, we obtain
\begin{align}
    \b{x}_{t+1} = \b{x}_{t} + \b{w}_{t}, \label{eqn:sys_model}
\end{align}
where $\b{w}_t$ is a correlated multivariate normal random vector for any given time-slot, but independent and identically distributed across $t$. In other words, $\b{w}_t \sim \mathcal{N}(0, Q)$, $ \forall t \in \N$, and $\b{w}_t \perp \b{w}_s, \forall t \neq s$, where $Q$ is the covariance matrix of $\b{w}_t$.

As discussed earlier, only one of the $M$ sensors can be requested to transmit state information in any given time-slot. Let $y_t \in \mathbb{R}$ denote the received sample at the monitor at time $t$. Suppose sensor $j$ transmitted its state information at time $t$. Then, we can write $y_t$ as follows 
\begin{equation}
    y_t = \b{c}^T_t\b{x}_t,
\end{equation}
where $\b{c}_t = \b{e}_j = [0, ..., 1, 0, ..., 0]^T$ and $\b{e}_j$ is the $j^{th}$ standard basis vector in $\R^M$.

At each time-slot, the monitor uses an estimator $\eta_t(\cdot)$ to map all previously received information, i.e. $y_t, y_{t-1},....,y_0$ to an estimate of the state $\b{x}_t$, denoted by $\b{\hat{x}}_t$.
\begin{equation}
    \label{eqn:estimator}
    \b{\hat{x}}_t = \eta_t(y_t, y_{t-1},....,y_0).
\end{equation}
We assume that the estimate $\hat{\b{x}}_0$ at time $t=0$ is known, with an error covariance matrix $P_0 = \E\left[\left(\b{x}_0 -\hat{\b{x}}_0\right)\left(\b{x}_0 -\hat{\b{x}}_0\right)^T\right]$. For example, if the state $\hat{\b{x}}_0 = \b{x}_0$ is known exactly at time $t=0$, then $P_0$ is the matrix with all zeroes.

\subsection{Setting up the Optimization Problem}
In this work, our objective is to design a scheduling and estimation policy that minimizes the long-term time-average of the estimation error at the monitor. More formally, we want to solve the following optimization problem
\begin{equation}
\begin{aligned}
\min_{\b{c}_t, \eta_t, t\in \N} \quad & \lim_{T \to \infty} \frac{1}{T} \E\left[\sum_{t=1}^T||\b{x}_t-\b{\hat{x}}_t||_2^2\right]\\
\textrm{s.t.} \quad & \b{x}_{t+1} = \b{x}_{t} + \b{w}_{t}, \forall t\\
  &y_t = \b{c}^T_t\b{x}_t, \forall t\\
  &\b{c}_t \in \{\b{e}_1,...,\b{e}_M \}, \forall t\\
  &\b{\hat{x}}_t = \eta_t(y_t,...,y_0), \forall t.
\end{aligned}\label{eqn:general_prob}
\end{equation}

Here, the sequence $\{\b{c}_t\}_{t=0}^\infty$ controls which sensor gets to transmit in each time slot and hence describes the scheduling policy. The sequence $\{\eta_t\}_{t=0}^\infty$ describes the estimation policy for $\hat{\b{x}}_t$.

Let $\pi(t) \in \{ 1,...,M\}$ denote the index of the sensor scheduled at time $t$. For a given scheduling $\pi$, it is well established \footnote{This is a classical result in estimation theory. A proof can be found in \cite[p. 143]{poor_textbook_1998}.} that the optimal estimator for \eqref{eqn:general_prob} is the Minimum Mean Square Error (MMSE) estimator, which is given as 
\begin{equation}\label{eqn:cond_exp}
    \begin{aligned}
        \b{\hat{x}}_t  &= \E[\b{x}_{t}| y_t, y_{t-1},...,y_{0}] =\E[\b{x}_{t}|x^{\pi(t)}_t, x^{\pi(t-1)}_{t-1},...,x^{\pi(0)}_{0}].
    \end{aligned}
    \end{equation}
        \section{Kalman Filter: The Optimal Estimator}\label{chap:kalmanfilter}
In general, (\ref{eqn:general_prob}) requires searching over the joint space of all causal scheduling and estimation policies, which suggests that an exact solution to the problem is intractable \cite{vasconcelos_observation-driven_2020}. However, we will establish that for a reasonably large class of scheduling policies, the overall optimization is separable and splits into separate estimation and scheduling problems. Further, the optimal estimation problem, i.e., equation \eqref{eqn:cond_exp}, can be obtained using the well known Kalman filter algorithm.

To start, we note that causal scheduling policies can be categorized into two classes: 
\begin{enumerate}
    \item \textit{Oblivious policies}: The scheduling decision at time $t$ can only be a function of previous scheduling decisions $\b{c}_0,...,\b{c}_{t-1}$ and previously received updates $y_0,...,y_{t-1}$. 
    \item \textit{Non-oblivious policies}: The scheduling decision at time $t$ can also be a function of the current state $\b{x}_t$.
\end{enumerate}
We are interested in the problem of centralized scheduling, where the base station decides which sensor to schedule at the beginning of each time slot. Since the base station makes the scheduling decision without access to the current state $\b{x}_t$ or the instantaneous errors at the sensors, \textit{throughout this work we will only consider oblivious and causal scheduling policies}.

Next, we describe the optimal MMSE estimator for our problem. Note that our problem can be viewed as a linear dynamical system (with a system matrix equaling the identity matrix) and time-varying observation matrices $\b{c}_t$. Given the sequence of observation matrices for a discrete time linear dynamical system with Gaussian noise, the MMSE can be obtained by running the Kalman filter algorithm\cite{kalman_new_1960, pei_elementary_2019}.

\subsection{Kalman Recursion}\label{chap:kalmanfilter-sec:recursion}
The Kalman filter computes the estimate of the state at time $t$ in a recursive manner using the estimate at time $t-1$ and the measurement at $t$. We describe the Kalman filter algorithm for our problem below. 

\label{lemma:Kalman_for_sys_model}
\noindent
    For a sequence of scheduling decisions $\b{c}_0,...,\b{c}_t$ and corresponding observations $y_0,...,y_t$ the optimal MMSE estimate $\b{\hat{x}}_t$ at time $t$ for problem \eqref{eqn:general_prob} is given by
    \begin{equation}
        \b{\hat{x}}_{t} = (I - K_t\b{c}_t^T)\b{\hat{x}}_{t-1} + K_t y_t,
    \end{equation}
    where $K_t$ is the Kalman gain matrix given by
    \begin{equation}
        K_t = \frac{P_{t}\b{c}_t}{\b{c}_t^TP_{t}\b{c}_t},
    \end{equation}
    and the matrix $P_t$ satisfies the following recursion
    \begin{equation}
    \label{eqn:Kalman_cov_eq}
        P_{t+1} = P_{t} - \frac{P_{t}\b{c}_t \b{c}_t^TP_{t}}{\b{c}_t^TP_{t}\b{c}_t} + Q, \forall t.
    \end{equation}

    The initial estimate $\b{\hat{x}}_0$ and the matrix $P_0$ are assumed to be known.

    Given the best MMSE estimate of the state at $t-1$, the probability distribution of the state at $t$ is a Gaussian distribution with the MMSE estimate as the mean, and a certain covariance. This is because of the Gaussian nature of $\b{w}_t$ and the linearity of the process. In other words, $\b{x}_t$ conditioned on all the measurements until $t-1$ is a multivariate normal random variable of the form $\b{x}_t \sim \mathcal{N}(\b{\hat{x}}_{t-1}, P_{t-1})$. Intuitively, the next best MMSE estimate would be a weighted average of the current best MMSE estimate and the new measurement at time $t$. The Kalman filter algorithm gives us the optimal way to construct these weights. 
    
    Refer to Appendix \ref{appendix:kalman_derivation} for the detailed derivation of the recursions for problem \eqref{eqn:general_prob}.

Observe that the estimate at $t$ is a linear function of the estimate at time $t-1$ and the observation made at time $t$. Further, note that the matrix $P_t$ in the Kalman recursion has a special property - it is the covariance matrix of the state estimate at time $t$ before the measurement at time $t$. Moreover, it is related to the covariance matrix of the random vector $\b{x}_t - \b{\hat{x}}_t$ as follows:
\begin{equation}\label{eqn:covariance_relation}
    P_{t+1} = \mathbb{E}\bigg[ \big(\b{x}_t - \b{\hat{x}}_t\big)\big(\b{x}_t - \b{\hat{x}}_t\big)^T \bigg] + Q.
\end{equation}

From the above equation, we observe that the diagonal elements of the matrix $P_{t+1} - Q$ represent the expected squared error between the actual state and the estimate of the $i$th process at time $t$, i.e.
\begin{equation}
    \E[(x^i_t - \hat{x}^i_t)^2] = (P_{t+1} - Q)_{ii}.
\end{equation}

Thus, the expected monitoring error at time $t$, under a scheduling policy $\pi \in \Pi$, can be obtained by running the Kalman filter recursion for the same policy, i.e.
\begin{equation}\label{eqn:expec_error_P_t_equivalence}
\begin{aligned}
    \E_{\pi}\left[||\b{x}_t - \b{\hat{x}}_{t}||^2_2\right] &= \sum_{i=1}^M\E_{\pi}\left[(x^i_t - \hat{x}^i_t)^2\right] = \E_\pi[\tr{P_{t+1}}] - \tr{Q},
\end{aligned}
\end{equation}
where $\Pi$ is the class of all oblivious scheduling policies. The expectation on the right-hand side indicates expectation over stochastic policies which allow randomness in scheduling.

\subsection{Decoupling of Estimation and Scheduling}
From the discussions in the previous subsection, we can establish a helpful result that lets us decouple the estimation and scheduling problem. 
\begin{lemma}\label{lemma:decoupling}
    When restricted to oblivious policies, the estimation and scheduling problem in \eqref{eqn:general_prob} can be decoupled.
\end{lemma}

Note that Lemma \ref{lemma:decoupling} requires the policy to be \textit{oblivious}. This is because equation \eqref{eqn:cond_exp} depends on $\pi(t)$, and if $\pi(t)$ depended on instantaneous error \textit{(non-oblivious policy)}, the scheduling and estimation problems could be coupled. For example, consider a scheduling policy which schedules the sensor which has maximum instantaneous error, i.e., $\pi(t) = \argmax_i |x^i_t - \hat{x}^i_t|$ at time $t$. Then, at time $t$, we know that the instantaneous errors of all the other sensors is bounded, i.e., $|x^{i}_t - \hat{x}^{i}_t| \leq |x^{\pi(t)}_t - \hat{x}^{\pi(t)}_t|$ for all $i$. The conditional expectation in \eqref{eqn:cond_exp} would exploit this information, hence coupling the scheduling decisions and the estimate $\hat{\b{x}}_t$. We now provide a simple proof for Lemma \ref{lemma:decoupling}.

\begin{proof}
     Considering only oblivious policies, the scheduling policy at $t$ is independent of  $\b{x}_t - \b{\hat{x}}_t$. It depends on the scheduling decisions and observations until $t-1$. In other words, the scheduling decision  $\pi(t)$ is fully determined at time $t-1$. We know that the optimal estimate at time $t$ is
    \begin{equation}\label{eqn:mmse}
        \b{\hat{x}}_t = \E[\b{x}_{t}|x^{\pi(t)}_t, x^{\pi(t-1)}_{t-1}, x^{\pi(t-2)}_{t-2},...,x^{\pi(0)}_{0}].
    \end{equation}
   Since $\pi(t)$ is decided \textit{before} time $t$, the sequence $\{\b{c}_k\}_{k=0}^{t}$ becomes known before time $t$, at time $t-1$. Once the measurement sequence is known \textit{a priori}, the MMSE takes the form of the Kalman Filter. Since this is true for any time $t$, the optimal estimator for oblivious policies is the Kalman Filter.  
\end{proof}
Observe that from the \textbf{orthogonality principle of the MMSE estimator}, all the information we can obtain about the process at time $k=t$ from the measurements from time $k=0$ to $k=t-1$ has been extracted by the MMSE estimator and is encoded in $\b{\hat{x}}_{t-1}$ and $P_{t-1}$.
 Even though oblivious policies can consider the realization of the measurements until $t-1$ to make the scheduling decision, there is no loss in optimality even if we do not consider them. 
 Observe that the recursion of $P_t$ as described in \eqref{eqn:Kalman_cov_eq} does not depend on the realization of the measurements, but only the scheduling decisions $\b{c}_t$. From the decoupling argument in Lemma \ref{lemma:decoupling} and the discussion above, we obtain the following equivalent optimization problem.
\noindent
    For the class of oblivious scheduling policies $\Pi$, the optimization problem \eqref{eqn:general_prob} is equivalent to
    \begin{equation}
    \begin{aligned}
    \min_{\pi \in \Pi} \quad & \lim_{T \to \infty} \frac{1}{T} \sum_{t=1}^T \E[\tr{P_t}]\\
    \textrm{s.t.}  \quad & P_{t+1} = P_t - \frac{P_t\b{c}_t\b{c}_t^TP_t}{\b{c}_t^TP_t\b{c}_t} + Q, \forall t,\\
      &\b{c}_t \in \{\b{e}_1,...,\b{e}_M \}, \forall t,
    \end{aligned}\label{eqn:kalman_prob_prior}
    \end{equation}
where $\pi$ denotes a scheduling policy, that involves choosing $\b{c}_t, \forall t \in \mathbb{N}$.


\subsection*{Remarks on Scheduling Policies}
Note that the Kalman recursion causes the mean estimation error $\tr{P_t}$ to depend on the entire history of scheduling decisions until time $t$. From Kalman filtering theory, we know that for optimal reconstruction of $\b{\hat{x}}_t$, we require information from all measurements up to time $t$. To solve (\ref{eqn:kalman_prob_prior}) for the finite horizon (i.e., finite $T$), the optimal schedule would be obtained from a dynamic program. However, this approach suffers from the curse of dimensionality, thus making it intractable to solve directly. 

One would naturally think of a policy which is greedy, i.e., schedule the sensor which minimizes the  error (i.e., $\tr{P_{t+1}}$) at each time step would be sufficiently close to optimal. However we show via simulations in chapter \ref{chap:simulations} that such a greedy policy is highly suboptimal. This is due to its myopic nature. 

It is interesting to note that $P_{t}$ does not converge when $t$ goes to infinity. This is because $\b{c}_t$ is a rank 1 matrix, and so, because of the dynamics of $P_t$, the quantity $P_{t+1} - P_t$ cannot be bounded in any normed sense. The long-term evolution of $P_t$ does not show any clear regularity. This is because the optimal MMSE estimator $\b{\hat{x}}_t = \E[\b{x}_{t}|x^{\pi(t)}_t, x^{\pi(t-1)}_{t-1}, x^{\pi(t-1)}_{t-1}, x^{\pi(t-2)}_{t-2},...,x^{\pi(0)}_{0}]$ cannot be truncated to a conditional expectation with finite history when the sensors are correlated, i.e. the entire sequence of scheduling decisions is necessary to construct the MMSE estimator at time $t$. This dependence on the entire scheduling history again makes the problem intractable. We solve this dilemma in the next section by designing lower and upper bounds on the monitoring error under any oblivious scheduling policy, using the AoIs of the sensors. This allows us to gain insight into the structure of the problem and propose efficient approximate solutions. 

	\section{Monitoring Error and Age-of-Information (AoI)}\label{chap:aoi-error-relation}
In this section, we derive lower and upper bounds for the monitoring error of each sensor using its AoI. 
Throughout, we assume that the covariance matrix $Q$ is invertible. This implies that $Q$ must also be strictly positive definite, i.e. $Q \succ 0$. The assumption is fairly general - it says that no sensor measurement can be viewed as a linear combination of other processes in the system. In other words, each of the $M$ processes has some noise innovation that is novel to itself. We will study the special case when $Q$ is non-invertible separately in chapter \ref{chap:scaling}.

 Let the age of information (AoI) of sensor $i$ be denoted by $h_i(t)$. Then, the AoI of source $i$ is defined as the time elapsed since the sensor $i$ was last scheduled. It evolves as follows 
\begin{equation}
\label{eqn:AoI_ev}
    h_i(t) = \begin{cases}
        h_i(t-1) + 1, \text{ if } \b{c}_t \neq \b{e}_i\\
        0, \text{ otherwise.}
    \end{cases}
\end{equation}
Further, let us denote the $ij^{th}$ element of the error covariance matrix $P_t$ as $p_{ij}(t)$, and the $ij^{th}$ element of the noise increment covariance matrix $Q$ as $q_{ij}$. Recall from our discussion in Section \ref{chap:kalmanfilter-sec:recursion} 
that the expected error in monitoring the $i$th process is related to the $i$th diagonal element of $P_{t+1}$, i.e.
\begin{equation}
    \E\left[(x^i_t - \hat{x}^i_t)^2\right] = p_{ii}(t+1) - q_{ii}.
\end{equation}

We will use this relationship between the monitoring error and elements of $P_t$ to upper bound the estimation error for each process with its weighted AoI. If we can successfully bound $p_{ii}(t)$, then we can bound the estimation error of the $i^{th}$ process.

\subsection{Bounds on Monitoring Error}
\begin{theorem}\label{thm:upper_bound}
The estimation error for process $i$ under a scheduling policy $\pi$, can be upper bounded using the AoI of the same process under the same scheduling policy as follows
\begin{align}\label{eqn:thm_upper_bound}
    \E_{\pi}\left[(x^i_t - \hat{x}^i_t)^2\right] \leq q_{ii}\E_\pi[h^\pi_i(t)].
\end{align}
Here $q_{ii}$ is the $i$th diagonal element of the matrix $Q$, and hence the noise variance of the $i^{th}$ process.
\end{theorem}
\begin{proof}
When sensor $i$ is scheduled, its error goes down to zero. For every subsequent time slot in which it is not scheduled, its expected error can grow by at most $q_{ii}$, which is the variance of the noise increment to process $i$. If there were other processes correlated with the $i$th process, then its expected error would have grown by something smaller than $q_{ii}$, since the base station would have received partial information about process $i$ when it received an update from a correlated sensor. Ignoring this reduction due to correlation, the expected error just grows as the AoI of the $i$th process (multiplied by the noise variance). The details are provided in Appendix \ref{appendix:thm_upper_bound_proof}.
\end{proof}
Note that we derived our upper bound by comparing the correlated case to the uncorrelated setting. Essentially, this upper bound tells us that the monitoring error of the $i^{th}$ source with correlation is no worse than when there is no correlation, a rather intuitive result. 


We next establish Theorem \ref{thm:lower_bound}, which states that the monitoring error of the $i^{th}$ sensor under any scheduling policy can be lower bounded by the product of the age of the $i^{th}$ sensor and an appropriate weight.
\begin{theorem}\label{thm:lower_bound}
The estimation error for process $i$ under a scheduling policy $\pi$, can be lower bounded using the AoI of the same process under the same scheduling policy as follows
\begin{align}\label{eqn:thm_lower_bound}
    \E_{\pi}\left[(x^i_t - \hat{x}^i_t)^2\right] \geq (q_{ii}-\b{q}_i^TQ_{-i}^{-1}\b{q}_i) \E_\pi[h^\pi_i(t)] \triangleq \Tilde{q}_{ii} \E_\pi[h^\pi_i(t)].
\end{align}
Here $q_{ii}$ is $i$th diagonal element of the matrix $Q$ i.e. the noise variance of the $i^{th}$ process, $\b{q}_i$ is the covariance of the noise of the $i^{th}$ process with the other processes, and $Q_{-i}$ is the noise covariance submatrix of the other processes.
\end{theorem}
\begin{proof}
To prove the lower bound, we calculate an upper bound for the maximum possible reduction in expected error due an update being received from sources other than $i$, in time-slots when $i$ is not scheduled. We do this by optimizing over the space of linear combinations of all sources other than $i$ and finding a linear combination that provides the maximum information regarding $i$. Of course, such a linear combination cannot actually be scheduled since we can only observe one dimension at a time. However, this leads to a lower bound since it overestimates the benefit due to correlation. The detailed proof of theorem \ref{thm:lower_bound} can be found in Appendix \ref{appendix:thm_lower_bound_proof}.
\end{proof}
Since $Q$ is positive definite, we know that $q_{ii} > 0, \forall i$ and $\Tilde{q}_{ii} = q_{ii} - \b{q}_i^TQ_{-i}^{-1}\b{q}_i > 0$ from the properties of positive definite matrices and their Schur complements \cite{boyd_block_nodate}. Thus, the lower bound derived above leads to a meaningful and non-trivial bound. It says that despite the presence of correlation, the expected error for any sensor under any oblivious policy is lower bounded by its corresponding weighted AoI. Thus, the expected error always lies between two weighted AoIs. In chapter \ref{chap:schedule_policy}, we will use this insight to design scheduling policies with constant factor optimality guarantees.

\subsection*{Discussion on Developed Bounds}
Note that in \eqref{eqn:thm_upper_bound}, the upper bound for source $i$ is the one-step variance of process $i$ times the age of sensor $i$. Whereas in \eqref{eqn:thm_lower_bound}, the weighting factor of the AoI of sensor $i$ is the maximum reduction in variance due to the correlation of source $i$ with the other sources subtracted from one-step variance of process $i$ itself. At a high level, while the upper bound is overly pessimistic regarding the benefits of correlation, the lower bound is overly optimistic. Moreover, in the special case when $Q$ is diagonal, i.e. the processes are uncorrelated, the two bounds match. 
\begin{figure}[!hbt]
    \centering
       
 \includegraphics[width=\linewidth]{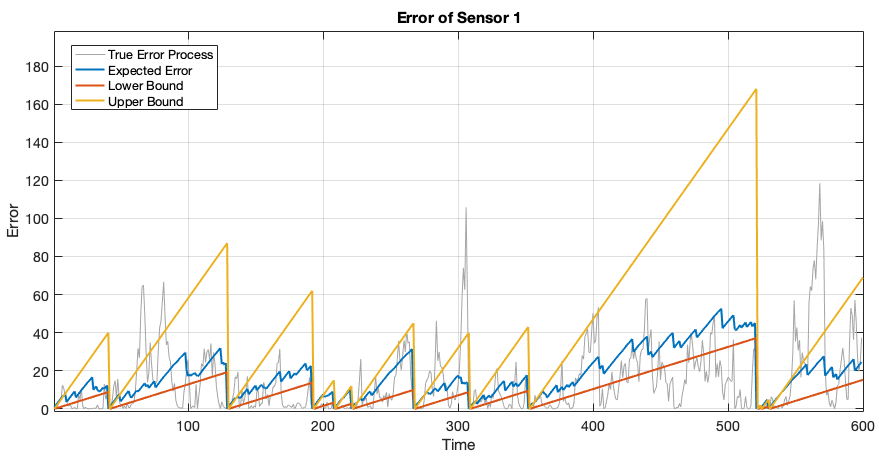}
  \caption{Upper and lower bounds demonstrated for $M = 20$ sensors. The scheduling process is the optimal stationary randomized schedule\cite{kadota_scheduling_2018}. The $Q$ matrix represents a highly correlated setting, where the diagonal elements are $q_{ii} = 1$, and the non-diagonal elements are $q_{ij} = 0.8$.}
  \label{fig:ub_lb_time_demo} 
\end{figure}

Figure \ref{fig:ub_lb_time_demo} demonstrates how the upper and lower bounds look for a setting with highly correlated sensors. Observe that the expected error process is fairly close to the lower bound obtained from Theorem \ref{thm:lower_bound}. But the lower bound is not achievable in general, because the ``maximum correlation" in a one-step update is overly optimistic given our communication constraints. Further, note that the true monitoring error is random, and oscillates around the expected error process. Throughout this work, we will focus on minimizing the expected error.

It is easy to verify that both the bounds are tight when $Q$ is diagonal. If sensor $i$ is uncorrelated with other sensors, then $\b{q}_i = 0$, and ${q}'_{ii} = q_{ii}$. Consequently, we have
\begin{equation*}
    \E_\pi[(x^i_t - \hat{x}^i_t)^2] = q_{ii}\E_\pi[h^\pi_i(t)].
\end{equation*}

        \section{Scaling of Optimal Monitoring Error}\label{chap:scaling}
In this section, we show that the estimation error of the optimal scheduling scheme for correlated sources is of the order of the square of the number of linearly independent processes being monitored. 

\subsection{The Invertible Case}\label{subsection:full_rank}
\begin{theorem}\label{thm:order_of_mag_full}
For an invertible covariance matrix $Q$, the monitoring error achieved by the optimal scheduling policy $\pi^*$ which minimizes (\ref{eqn:kalman_prob_prior}) satisfies
    \begin{equation}\label{eqn:p_opt_lb_and_ub}
    \begin{aligned}
        \frac{1}{2}\left(\sum_{i=1}^M\sqrt{\Tilde{q}_{ii}}\right)^2 - \frac{1}{2}\sum_{i=1}^M \Tilde{q}_{ii}\leq P_{OPT} \leq \left(\sum_{i=1}^M\sqrt{q_{ii}}\right)^2-\sum_{i=1}^Mq_{ii}.
    \end{aligned}
    \end{equation}
    Here $P_{OPT}$ is the monitoring error under the policy $\pi^*$.
\end{theorem}

Theorem \ref{thm:order_of_mag_full} tells us that the monitoring error of the optimal schedule for the case of invertible $Q$ matrices roughly scales as $M^2$ (whenever $q_{ii}$ and $\Tilde{q}_{ii}$ are independent of $M$).
When $\text{rank}(Q) = M$, we can use the results of theorem \ref{thm:upper_bound} and \ref{thm:lower_bound} to bound performance of the optimal scheduling policy. 

\begin{proof}[Proof of Theorem \ref{thm:order_of_mag_full}]
Note that from theorem \ref{thm:upper_bound} and \ref{thm:lower_bound}, for any sequence of schduling decisions $\pi(t)$, we have for every $i = 1,..., M$,
\begin{equation}
    \Tilde{q}_{ii}h^\pi_i(t) \leq \E[(x^i_t -\hat{x}^i_t)^2|\pi] \leq q_{ii}h^\pi_i(t) \label{eqn:thm_aoi_lb_ub_res}
\end{equation}
where the notation is as defined in the previous section. Hence,
\begin{equation}
     \sum_{i=1}^M\Tilde{q}_{ii}h^\pi_i(t) \leq \sum_{i=1}^M\E[(x^i_t -\hat{x}^i_t)^2|\pi] \leq \sum_{i=1}^Mq_{ii}h^\pi_i(t).
\end{equation}
where we have made the dependence on the policy clear. Consequently, 
\begin{align}
    \lim_{T\to\infty}\frac{1}{T}\sum_{t=1}^T\sum_{i=1}^M \E[(x^i_t -\hat{x}^i_t)^2] &\geq \lim_{T\to\infty}\frac{1}{T}\sum_{t=1}^T\sum_{i=1}^M\E[\Tilde{q}_{ii}h^\pi_i(t)]\label{eqn:opt_full_rank_lb}\\
    \lim_{T\to\infty}\frac{1}{T}\sum_{t=1}^T\sum_{i=1}^M\E[(x^i_t -\hat{x}^i_t)^2] &\leq \lim_{T\to\infty}\frac{1}{T}\sum_{t=1}^T\sum_{i=1}^M\E[q_{ii}h^\pi_i(t)]\label{eqn:opt_full_rank_ub}.
\end{align}
Let $\pi^* \in \Pi$ be the optimal scheduling policy for problem (\ref{eqn:kalman_prob_prior}). To obtain a lower bound, we employ (\ref{eqn:opt_full_rank_lb}) to get
\begin{align}
    \min_{\pi \in \Pi}\lim_{T\to\infty}&\frac{1}{T}\sum_{t=1}^T\sum_{i=1}^N\E[(x^i_t -\hat{x}^i_t)^2]\nonumber\\
    &\geq \lim_{T\to\infty}\frac{1}{T}\sum_{t=1}^T\sum_{i=1}^M\E[\Tilde{q}_{ii}h^{\pi^*}_i(t)]\nonumber\\
    &\geq \min_{\pi' \in \Pi}\lim_{T \to \infty}\frac{1}{T}\sum_{t=1}^T\sum_{i=1}^M\E[\Tilde{q}_{ii}h^{\pi'}_i(t)]\nonumber\\
    &\stackrel{(a)}{=} \min_{\pi' \in \Pi}\lim_{T \to \infty}\frac{1}{T}\sum_{t=1}^T\sum_{i=1}^M\E[\Tilde{q}_{ii}(h^{\pi'}_i(t)+1)] - \sum_{i=1}^M \Tilde{q}_{ii}\nonumber\\
    &\stackrel{(b)}{\geq} \frac{1}{2}\left(\sum_{i=1}^M\sqrt{\Tilde{q}_{ii}}\right)^2+\frac{1}{2} \sum_{i=1}^M \Tilde{q}_{ii} -\sum_{i=1}^M\Tilde{q}_{ii}\nonumber\\
    &= \frac{1}{2}\left(\sum_{i=1}^M\sqrt{\Tilde{q}_{ii}}\right)^2-\frac{1}{2} \sum_{i=1}^M \Tilde{q}_{ii}\label{eqn:lb_final_full_rank},
\end{align}
where inequality $(b)$ comes from the universal lower bound of all weighted age of information \cite[Theorem 6]{kadota_scheduling_2018}. The algebraic manipulation in equality $(a)$ is required to employ \cite[Theorem 6]{kadota_scheduling_2018} as their work exploits the fact that AoI is lower bounded by 1.

\noindent
For the upper bound, we employ (\ref{eqn:opt_full_rank_lb}), and note that,
\begin{align}
     \min_{\pi \in \Pi}\lim_{T \to \infty}\frac{1}{T}\sum_{t=1}^T&\sum_{i=1}^M \E[(x^i_t -\hat{x}^i_t)^2]\nonumber\\
     &\leq \lim_{T \to \infty}\frac{1}{T}\sum_{t=1}^T\sum_{i=1}^M\E(x^i_t -\hat{x}^i_t)^2] && \forall \pi \in \Pi, \nonumber\\
     &\leq \lim_{T \to \infty}\frac{1}{T}\sum_{t=1}^T\sum_{i=1}^M\E[q_{ii}h^\pi_i(t)] &&\forall \pi \in \Pi. 
\end{align}
Therefore, if we choose the optimal stationary random policy on the RHS (as defined in \cite{kadota_scheduling_2018}), we get,
\begin{align}
    \min_{\pi \in \Pi}\lim_{T \to \infty}\frac{1}{T}\sum_{t=1}^T\sum_{i=1}^M\E[(x_t^i - \hat{x}^i_t)^2] \leq  \left(\sum_{i=1}^M \sqrt{q_{ii}}\right)^2 - \sum_{i=1}^Mq_{ii}.\label{eqn:ub_final_full_rank}
\end{align}
From equations (\ref{eqn:lb_final_full_rank}) and (\ref{eqn:ub_final_full_rank}), it is easy to see that,
\begin{align}
\frac{1}{2}\left(\sum_{i=1}^M\sqrt{\Tilde{q}_{ii}}\right)^2 - \frac{1}{2}\sum_{i=1}^M\Tilde{q}_{ii}\leq P_{OPT} \leq \left(\sum_{i=1}^M\sqrt{q_{ii}}\right)^2 - \sum_{i=1}^Mq_{ii}.
\end{align}
\end{proof}

\subsection{The Non-Invertible Case}\label{subsection:low_rank}
Now, we bound the performance of optimal policies when some sources are linearly dependent on other sources. We prove a rather intuitive result: the monitoring error of the optimal schedule roughly scales as the square of the dimensionality of the random processes. 
\begin{theorem}\label{thm:order_of_mag_low}
Given a covariance matrix $Q$ with rank $L < M$, the optimal scheduling policy $\pi^*$  which minimizes (\ref{eqn:kalman_prob_prior}) achieves
    \begin{equation}\label{eqn:p_opt_low_rank}
    \begin{aligned}
        \frac{L(L+1)}{2}\lambda_L(Q') - \tr{Q} \leq P_{OPT}
        \leq c\left(\left(\sum_{i=1}^L \sqrt{q'_{ii}}\right)^2 - \sum_{i=1}^Lq'_{ii}\right)
    \end{aligned}
    \end{equation}
    where $Q'$ is a rank $L$ submatrix contained in $Q$, $P_{OPT}$ is the estimation error of the policy $\pi^*$, and $\lambda_{L}(Q')$ is the $L^{th}$-largest eigenvalue of the matrix $Q'$, i.e. the smallest eigenvalue of $Q'$.
\end{theorem}
\begin{proof}
    The proof for this result involves three key steps. First, we lower bound the optimal monitoring error under a given covariance matrix $Q$ by considering a relaxation of the scheduling problem, in which the scheduler is allowed to schedule using $\b{c}_t \in \mathbb{R}^{M}$ instead of $\b{c}_t \in \{\b{e}_1,...,\b{e}_M \}$. Second, we show that the relaxed problem can then be converted to an equivalent \textit{full-rank} relaxed problem involving the monitoring of only $L$ sensors, instead of the original $M$. Finally, we show that the monitoring error for the full-rank relaxed problem involving $L$ sensors scales as $O(L^2)$. The detailed proof can be found in the Appendix \ref{appendix:thm_order_of_mag_low}.
\end{proof}

Even though the inequality \eqref{eqn:p_opt_low_rank} looks rather complicated, it is intuitive for most low-rank covariance matrices. Consider a family of covariance matrices where the non-zero eigenvalues are not too skewed, i.e., they are close to each other in magnitude. In this case, Theorem \ref{thm:order_of_mag_low} states that the optimal time average monitoring error scales roughly as $L^2$. Essentially, access to $L \leq M$ sensors (of which $M-L$ are linear combinations of the other sensors) will not give any order improvement in error performance. 
This scaling is similar to the result we obtained for the invertible case.

        \section{Design of Scheduling Policies}\label{chap:schedule_policy}
In this chapter, we use the bounds derived in section \ref{chap:aoi-error-relation} to propose Max-Weight style scheduling policies to minimize error with performance guarantees. We also establish that AoI is a good metric for error minimization even in correlated settings. Max-Weight style policies arise from the greedy one-step minimization of a well-chosen Lyapunov function. In the following subsections, we propose two Lyapunov functions and derive policies from the greedy minimization of their one-time-slot drift.

\subsection{Maximum Expected Error (MEE) policy}\label{chap:schedule_policy-section:mee}
In this section, we develop a scheduling policy that loosely schedules the sensor that has accumulated the maximum expected error until time $t$.
Recall from Chapter \ref{chap:kalmanfilter} that our original optimization problem \eqref{eqn:general_prob} is equivalent to the reformulated problem \eqref{eqn:kalman_prob_prior}. The objective function in \eqref{eqn:kalman_prob_prior} is the time average of the quantity $\tr{P_t}$. To optimize this quantity, we consider the following Lyapunov function
\begin{equation}
    L(t) \stackrel{\Delta}{=} \sum_{i=1}^M\frac{p_{ii}(t)}{\sqrt{q_{ii}}}.
\end{equation}
Note that the above Lyapunov function is a weighted sum of the expected errors of each process (before measurement).

Intuitively, in order to minimize the time average of $\tr{P_t}$, we would want $L(t)$ to be ``small" at all times. We define the Lyapunov drift under a policy $\pi$ as $\Delta_\pi(t) := L(t+1) - L(t)$. We define the state of the network at time $t$ to be the covariance matrix $P_t$. Notice that the expected Lyapunov drift, given the network state at time $t$, is
\begin{align*}
    \E[\Delta_\pi | P_t] &= \E[L(t+1) - L(t)| P_t]\\
    &= \sum_{i=1}^M\frac{\E\left[p_{ii}(t+1) - p_{ii}(t)| P_t\right]}{\sqrt{q_{ii}}}
\end{align*}

Using the Kalman recursion, we can upper bound the drift in each time-slot as follows
\begin{align}
    \E&[\Delta_\pi(t) | P_t]\nonumber\\
    &= \sum_{i=1}^M\frac{1}{\sqrt{q_{ii}}}\E\bigg[p_{ii}(t) - \sum_{l=1}^M\frac{p_{il}^2(t)}{p_{ll}(t)}u^\pi_{l}(t) + q_{ii} - p_{ii}(t) \bigg| P_t\bigg]\nonumber\\
    &= \sum_{i=1}^M\frac{1}{\sqrt{q_{ii}}}\E\left[q_{ii} - \sum_{l=1}^M\frac{p_{il}^2(t)}{p_{ll}(t)}u^\pi_{l}(t) \bigg| P_t\right]\nonumber\\
     &= \sum_{i=1}^M\sqrt{q_{ii}} - \sum_{i=1}^M\sum_{l=1}^M\frac{1}{\sqrt{q_{ii}}}\E\left[\frac{p_{il}^2(t)}{p_{ll}(t)}u^\pi_{l}(t) \bigg| P_t\right]\nonumber\\
    &\leq \sum_{i=1}^M\sqrt{q_{ii}} - \sum_{i=1}^M\frac{p_{ii}(t)}{\sqrt{q_{ii}}}\E[u^\pi_i(t) | P_t], \label{eqn:mee_drift_ub}
\end{align}
where the last inequality is obtained by only considering the terms when $l=i$ in the summation, as $p_{il}^2(t)/p_{ll}(t)$ is always positive.
Recall that $u_i(t) = 1$ for only one value $i \in \{1,...,M\}$, and $u_j(t) = 0$ for $i\neq j$. It turns out that the policy that greedily minimizes the upper bound of the above Lyapunov drift achieves certain performance guarantees. The upper bound \eqref{eqn:mee_drift_ub} is minimized when 
\begin{equation}
\boxed{
    \pi^{MEE}(t) = \argmax_i \left( \frac{p_{ii}}{\sqrt{q_{ii}}}\right).}
\end{equation}
We define this policy to be the \textbf{Maximum Expected Error (MEE) Policy}. We further define $P_{MEE}$ as the time average expected error when the MEE policy is employed, i.e.,
\begin{equation}
    P_{MEE} \stackrel{\Delta}{=} \lim_{T\to\infty}\frac{1}{T}\sum_{t=1}^T\E\left[||\b{x}_t - \hat{\b{x}}_t||_2^2\right]
\end{equation}
for the MEE policy. 
We now provide the performance guarantee for this policy.

\begin{theorem}[Performance guarantee for the MEE Policy]\label{thm:mee_policy}
The maximum expected error (MEE) scheduling policy is a constant factor away from optimality. In particular,
\begin{align}
\frac{P_{MEE}}{P_{OPT}} &\leq 2\frac{\left(\sum_{i=1}^M\sqrt{q_{ii}}\right)^2 - \sum_{i=1}^M q_{ii}}{ \left(\sum_{i=1}^M\sqrt{\Tilde{q}_{ii}}\right)^2 - \sum_{i=1}^M \Tilde{q}_{ii}}
\end{align}
where $P_{OPT}$ is the optimal time average expected error over all oblivious scheduling policies, and $\Tilde{q}_{ii} = q_{ii} - \b{q}_i^TQ_{-i}^{-1}\b{q}_i$ as described in the previous section.
\end{theorem}

\begin{proof}
The proof involves two parts. First, we upper bound $P_{MEE}$ by upper bounding the drift in every time-slot using the stationary randomized policy proposed in \cite{kadota_scheduling_2018}. To do so, we set the weights in the corresponding weighted-AoI problem to be the variances, i.e. $q_{ii}$. Second, we develop a lower bound for $P_{OPT}$ in Chapter \ref{chap:scaling}, Section \ref{subsection:full_rank}. Putting the two bounds together we get the result above. The detailed proof can be found in Appendix \ref{appendix:mee_proof}.
\end{proof}

\subsection{Maximum Weighted Age (MWA) policy}\label{chap:schedule_policy-section:mwa}
In this section, we develop a scheduling policy that schedules the sensor that has accumulated the maximum weighted age until time $t$,  and prove its performance guarantees. This policy ignores the effects of correlation for scheduling. Consider the Lyapunov function
\begin{equation}
    L(t) \stackrel{\Delta}{=} \sum_{i=1}^M\sqrt{q_{ii}}~h_i(t-1).
\end{equation}
The above Lyapunov function is a weighted sum of the ages of each process at time $t-1$, before the scheduling decision is made at time $t$. Intuitively, from the discussion in Chapter \ref{chap:aoi-error-relation}, in order to minimize the time average error, we would want the average age of all sources to be small and hence, $L(t)$ to be ``small" at all times. The Lyapunov drift for policy $\pi$ is given as $\Delta_\pi(t) \stackrel{\Delta}{=} L(t+1) - L(t)$. We define the state of the network at time $t$ to be the age vector $\b{h}_t = [h_1(t-1), h_2(t-1),...,h_M(t-1)]$. Using the evolution of AoI \eqref{eqn:AoI_ev}, we can calculate the expected drift as follows.
\begin{align}
    \E[\Delta_\pi(t) &| \b{h}_t] = \E[L(t+1) - L(t)| \b{h}_t]\nonumber\\
    &= \sum_{i=1}^M\sqrt{q_{ii}}~\E\left[h_{i}(t) - h_{i}(t-1)| \b{h}_t\right]\nonumber\\
    &= \sum_{i=1}^M\sqrt{q_{ii}}\bigg(1 - (h_{i}(t-1)+1)\E\left[u^\pi_i(t)| \b{h}_t\right] \bigg).\label{eqn:mwa_drift}
\end{align}
where the last equality is obtained from the dynamics of the age of information with the scheduling constraint in place.
Recall that $u_i(t) = 1$ for only one value $i \in \{1,...,M\}$, and $u_j(t) = 0$ for $i\neq j$. Equation \eqref{eqn:mwa_drift} is minimized when 
 \begin{equation}
    \boxed{\pi^{MWA}(t) = \argmax_i \bigg( \sqrt{q_{ii}} ~h_{i}(t-1) \bigg).}
\end{equation}
We define this policy to be the \textbf{Maximum Weighted Age (MWA) Policy}. We further define $P_{MWA}$ as the time average expected error when the MWA policy is employed, i.e.,
\begin{equation}
    P_{MWA} \stackrel{\Delta}{=} \lim_{T\to\infty}\frac{1}{T}\sum_{t=1}^T\E_{}\left[||\b{x}_t - \hat{\b{x}}_t||_2^2\right].
\end{equation}
We now provide performance guarantees for the MWA policy.

\begin{theorem}[Performance guarantee for the MWA Policy]\label{thm:mwa_policy}
The Maximum Weighted Age (MWA) scheduling policy is a constant factor away from optimality. In particular,
\begin{align}
\frac{P_{MWA}}{P_{OPT}} &\leq 2\frac{\left(\sum_{i=1}^M\sqrt{q_{ii}}\right)^2 - \sum_{i=1}^Mq_{ii}}{ \left(\sum_{i=1}^M\sqrt{\Tilde{q}_{ii}}\right)^2 - \sum_{i=1}^M \Tilde{q}_{ii}}
\end{align}
where $P_{OPT}$ is the optimal time average expected error over all oblivious scheduling policies, and $\Tilde{q}_{ii} = q_{ii} - \b{q}_i^TQ_{-i}^{-1}\b{q}_i$ as described in the previous section.
\end{theorem}

\begin{proof}
The proof is similar to that of Theorem \ref{thm:mee_policy}. We create an upper bound using the stationary randomized policy for a corresponding weighted-AoI setting and use the lower bound derived in Chapter \ref{chap:scaling}, Section \ref{subsection:full_rank}. The detailed proof is provided in \ref{appendix:mwa_proof}.
\end{proof}

The MWA policy ignores the off-diagonal elements of $Q$ and only uses the diagonal elements, which is the variance of each process $i$. Thus, the MWA policy also ignores the correlation and treats the processes as uncorrelated for the purposes of scheduling. Despite this, it obtains the same performance guarantee as the MEE policy, which requires the scheduler to know the Kalman covariance matrices.

Note, however, that the performance guarantees for the MWA policy rely on running the Kalman filter recursion for estimation, which exploits the correlation between the processes to get the best quality estimates. \textit{Thus, while correlation does not play a very important role in scheduling, it plays a crucial role in improving the quality of estimation.} In Chapter \ref{chap:simulations}, we study the performance of these algorithms empirically.

Another interesting observation is that when all the sensors are uncorrelated, then $\Tilde{q}_{ii} = q_{ii}$ and both the MEE and MWA policy become factor $2$ optimal. This is expected from the results of \cite{sun_age_2019}, as for the uncorrelated case, the problem defined in (\ref{eqn:kalman_prob_prior}) becomes a weighted-sum of AoI minimization problem, for which the MWA policy is known to be a factor of 2 optimal \cite{sun_age_2019}.



        \section{Extensions to General $A$ Matrices}\label{chap:extension}
In this chapter, we extend the results developed in chapters \ref{chap:aoi-error-relation}, \ref{chap:scaling}, and \ref{chap:schedule_policy} to a more relaxed system model. In chapter \ref{chap:intro}, we described the system model of a discrete Wiener process with correlated noise. Now, we extend this setting to a system with a diagonal system matrix $A$. Recall that the most general linear dynamical system with Gaussian innovations can be written as 
\begin{equation}
    \b{x}_{t+1} = A \b{x}_t + \b{w}_t.
\end{equation}
In equation \eqref{eqn:sys_model}, $A = I$. This is the discrete Wiener process. In this chapter we let $A$ be a diagonal matrix with the diagonal elements $a_{ii} \geq 1$. In particular, we will deal with systems of the form
\begin{equation}
    \b{x}_{t+1} = \begin{bmatrix}
        a_{11} & & \\
        & \ddots & \\
        & & a_{MM}
    \end{bmatrix} \b{x}_t + \b{w}_t,\label{eqn:sys_model_A}
\end{equation}
where $\b{w}_t$ is an iid zero mean Gaussian random vector such that $\E[\b{w}_t\b{w}_t^T] = Q$. Like in the previous setting, only one of the $M$ sensors can be requested to transmit state information in any given time-slot. Let $y_t \in \mathbb{R}$ denote the received sample at the monitor at time $t$. Then,
\begin{equation}
    y_t = \b{c}^T_t\b{x}_t,
\end{equation}
where $\b{c}_t = \b{e}_j = [0, ..., 1, 0, ..., 0]^T$ and $\b{e}_j$ is the $j^{th}$ standard basis vector in $\R^M$.

\subsection{Monitoring Error and Age-of-Information with Diagonal $A$ Matrices}
In this section, we extend the results of chapter \ref{chap:aoi-error-relation} to systems with diagonal $A$ matrices. In particular, we shall upper and lower bound the monitoring error in systems with diagonal $A$ matrices with the respective AoI of each sensor.  We will then use the developed bounds to design certain heuristic scheduling sensors to provide low monitoring error.
The intuition behind the developed bounds is straightforward and is very similar to the monitoring of a discrete Wiener process.

\begin{theorem}[Extension of Theorem \ref{thm:upper_bound}]\label{thm:ext_A_upper_bound}
The estimation error for process $i$, as described in \eqref{eqn:sys_model_A}, under a scheduling policy $\pi$, can be upper bounded using the AoI of the same process under the same scheduling policy as follows
\begin{align}\label{eqn:thm_upper_bound_A}
    \E_{\pi}\left[(x^i_t - \hat{x}^i_t)^2\right] &\leq q_{ii}\E_\pi\left[\sum_{k=0}^{h^\pi_i(t)}a^{2k}_{ii}\right]\nonumber\\
    &= \begin{cases}
       q_{ii}\E_\pi[h^\pi_i(t)]  & a_{ii} = 1\\
    q_{ii}\E_\pi\left[\frac{a_{ii}^{2h^\pi_i(t)} - 1}{a_{ii}^2 - 1}\right]  & a_{ii} > 1
    \end{cases}. 
\end{align}
Here $q_{ii}$ is the $i$th diagonal element of the matrix $Q$, and hence the noise variance of the $i^{th}$ process.
\end{theorem}
\noindent
The proof can be found in Appendix \ref{appendix:extension_upper_bound}.

\begin{theorem}[Extension of Theorem \ref{thm:lower_bound}]\label{thm:ext_A_lower_bound}
The estimation error for process $i$, as described in \eqref{eqn:sys_model_A}, under a scheduling policy $\pi$, can be lower bounded using the AoI of the same process under the same scheduling policy as follows
\begin{align}\label{eqn:thm_lower_bound_A}
    \E_{\pi}\left[(x^i_t - \hat{x}^i_t)^2\right] &\geq \Tilde{q}_{ii}\E_\pi\left[\sum_{k=0}^{h^\pi_i(t)}a^{2k}_{ii}\right]\nonumber\\ 
    &= \begin{cases}
       \Tilde{q}_{ii}\E_\pi[h^\pi_i(t)]  & a_{ii} = 1\\
    \Tilde{q}_{ii}\E_\pi\left[\frac{a_{ii}^{2h^\pi_i(t)} - 1}{a_{ii}^2 - 1}\right]  & a_{ii} > 1
    \end{cases}. 
\end{align}
Here $q_{ii}$ is $i$th diagonal element of the matrix $Q$ i.e. the noise variance of the $i^{th}$ process, $\b{q}_i$ is the covariance of the noise of the $i^{th}$ process with the other processes, and $Q_{-i}$ is the noise covariance submatrix of the other processes, and $\Tilde{q}_{ii} \triangleq q_{ii} - \b{q}_i^{T}Q_{-i}^{-1}\b{q}_i$.
\end{theorem}
\noindent
The proof can be found in Appendix \ref{appendix:extension_lower_bound}.
\subsection{Scheduling Policy Design for Diagonal $A$ Matrix}\label{sec:extension_schedule}
In this section, we develop a Whittle-Index-like policy for monitoring systems with diagonal $A$ matrices. From the result of theorem \ref{thm:ext_A_upper_bound} (equation \eqref{eqn:thm_upper_bound_A}), we have the following relation for the monitoring error and age of information (for $a_{ii} > 1)$.
\begin{equation}
     \E_{\pi}\left[(x^i_t - \hat{x}^i_t)^2\right] \leq \E_\pi\left[q_{ii}\frac{a_{ii}^{2h^\pi_i(t)} - 1}{a_{ii}^2 - 1}\right]. 
\end{equation}
Summing over all sensors and over time, we get
\begin{align}
     \frac{1}{T}\sum_{t=1}^T\sum_{i=1}^M\E_{\pi}&\left[(x^i_t - \hat{x}^i_t)^2\right]\nonumber\\
     &\leq \frac{1}{T}\sum_{t=1}^T\sum_{i=1}^M\E_\pi\left[q_{ii}\frac{a_{ii}^{2h^\pi_i(t)} - 1}{a_{ii}^2 - 1}\right]. \label{eqn:ext_A_prob_penultimate}
\end{align}
Let us define $f_i(\cdot): \N\cup\{0\} \to [0,\infty)$ as follows:
\begin{equation}
    f_i(x) = q_{ii}\frac{a_{ii}^{2x} - 1}{a_{ii}^2 - 1}.
\end{equation}
Note that $f_i(x)$ is a non-decreasing function of $x$. Using this definition, we use \eqref{eqn:ext_A_prob_penultimate}, as $T \to \infty$ as
\begin{align}
     \lim_{T \to \infty}\frac{1}{T}\sum_{t=1}^T\sum_{i=1}^M&\E_{\pi}\left[(x^i_t - \hat{x}^i_t)^2\right]\nonumber\\ 
     &\leq \lim_{T \to \infty}\frac{1}{T}\sum_{t=1}^T\sum_{i=1}^M\E_\pi\left[f_{i}(h^\pi_i(t))\right]. \label{eqn:ext_A_upper_bound_problem}
\end{align}

In \cite{tripathi2019functionsofage, tripathi2019whittle}, the authors develop a Whittle-Index policy to minimize the time average sum of general functions of AoI. Note that in equation \eqref{eqn:ext_A_upper_bound_problem}, we have a time average sum of non-decreasing functions of age on the right-hand side. If we minimize the time average of the function of age, we get a policy that upper bounds the expected monitoring error. According to \cite{tripathi2019whittle}, the Whittle index policy for the RHS of the equation \eqref{eqn:ext_A_upper_bound_problem} is
\begin{align}
    \pi(t) &= \argmax_{i}\left\{ (h_i(t-1)+1)f_{i}(h_i(t-1) + 1) - \sum_{k=0}^{h_i(t-1)}f_{i}(k)\right\}\nonumber\\
    &= \argmax_{i}\bigg\{ (h_i(t-1)+1)q_{ii}\frac{a_{ii}^{2h_i(t-1) + 2} - 1}{a_{ii}^2 - 1}\nonumber\\ 
    &\qquad \qquad \qquad \qquad \qquad \qquad \qquad - \frac{q_{ii}}{a^2_{ii} - 1}\sum_{k=0}^{h_i(t-1)}(a_{ii}^{2k} - 1)\bigg\}
\end{align}
Hence, we propose the Whittle Index for Estimated Error (WI-EE) policy as follows:

\begin{framed}
\begin{align}
    \pi^{WI-EE}(t) &= \argmax_{i}\bigg\{ (h_i(t-1)+1)q_{ii}\frac{a_{ii}^{2h_i(t-1) + 2} - 1}{a_{ii}^2 - 1}\nonumber\\ 
    &\qquad \qquad - \frac{q_{ii}}{a^2_{ii} - 1}\sum_{k=0}^{h_i(t-1)}(a_{ii}^{2k} - 1)\bigg\}
\end{align}
\end{framed}
Note that when $a_{ii} = 1$ for all $i$, then we get back the Max-Weight for Age (MWA) policy (as defined in Section \ref{chap:schedule_policy-section:mwa}) from the WI-EE policy.
\subsection{Lower Bound on Optimal Monitoring Error}\label{sec:lb_extension}
In this section, we provide a universal lower bound (like in Section \ref{subsection:full_rank} and \ref{subsection:low_rank}) for systems with general $A$ matrices. For any $A$ matrix, we lower-bound the estimation independent of the scheduling policy using Weyl's inequalities. 
\begin{theorem}\label{thm:lower_bound_for_general_A}
    The time average estimation error for system in \eqref{eqn:intro_sys_model} (with $A \succ I$) for any scheduling policy is universally lower bounded as 
    \begin{align}\label{eqn:lower_bound_general_A}
        &\lim_{T \to \infty}\frac{1}{T}\sum_{t=1}^T\sum_{i=1}^M \E[(\b{x}_t^i - \hat{\b{x}}^i_t)^2] \geq \nonumber\\
         &\frac{\lambda_{M}(AA^T)^{M+1} - (M+1)\lambda_{M}(AA^T) - M}{\lambda_1(AA^T)(\lambda_{M}(AA^T) - 1)^2}\lambda_{M}(Q) -\frac{\tr{Q}}{\lambda_1(AA^T)},
    \end{align}
    where $\lambda_M(\cdot)$ denotes the smallest eigenvalue and $\lambda_1(\cdot)$ denotes the largest eigenvalue.
\end{theorem}
\begin{proof}[Proof of Theorem \ref{thm:lower_bound_for_general_A}]
    For this proof, we make use of the \textit{apriori} covariance matrix $P_{t|t-1}$ which we will denote as $P_t$ as shorthand, as done in other proofs. We establish the following lemma which will help us prove the result.

    \begin{lemma}\label{lemma:eig_p_t_gen_A}
    There exists a time $t_0$ after which the $i^{th}$ largest eigenvalue of $P_t$ denoted $\lambda_i(P_t)$ satisfies the following inequality
    \begin{align}
        \lambda_i(P_t) \geq \frac{\lambda_M(AA^T)^{M-i+1} - 1}{\lambda_M(AA^T) - 1}\lambda_M(Q)
    \end{align}
    for all time $t > t_0 > 0$.
    \end{lemma}
    The proof involves using Weyl's inequalities and multiplicative matrix perturbation results. The detailed proof of this lemma can be found in Appendix \ref{appendix:eigenvalue_inequality}. We now relate the expected monitoring error to the trace of the \textit{apriori} covariance matrix. For the case with general $A$ matrices, equation \eqref{eqn:covariance_relation} holds true with some modification. The Kalman filter recursion gives us the following relation,
    \begin{equation}
        P_{t+1} = A\E\left[(\b{x}_t - \hat{\b{x}}_t)(\b{x}_t - \hat{\b{x}}_t)^T\right]A^T + Q.
    \end{equation}
    From the above inequality, it is immediate that
    \begin{equation}
       \tr{A^TA\E\left[(\b{x}_t - \hat{\b{x}}_t)(\b{x}_t - \hat{\b{x}}_t)^T\right]} = \tr{P_{t+1}} - \tr{Q}.
    \end{equation}
    Note that, the LHS is bounded above as 
    \begin{align}
         \text{tr}\bigg(A^TA\E&\left[(\b{x}_t - \hat{\b{x}}_t)(\b{x}_t - \hat{\b{x}}_t)^T\right]\bigg)\nonumber\\
         &\leq \lambda_1(A^TA)\tr{\E\left[(\b{x}_t - \hat{\b{x}}_t)(\b{x}_t - \hat{\b{x}}_t)^T\right]} \nonumber\\
         &= \lambda_1(AA^T)\sum_{i=1}^M\E[(x^i_t - \hat{x}^i_t)^2].
    \end{align}
Consequently,
    \begin{align}
        \lambda_{1}(AA^T)&\sum_{i=1}^M\E[(x^i_t - \hat{x}^i_t)^2] \nonumber\\
        &\geq \tr{P_{t+1}} - \tr{Q} \nonumber\\
        &= \sum_{i=1}^M \lambda_i(P_{t+1}) - \tr{Q}\nonumber\\
        &\geq \sum_{i=1}^M\frac{\lambda_M(AA^T)^{M-i+1} - 1}{\lambda_M(AA^T) - 1}\lambda_M(Q) - \tr{Q}.
    \end{align}
    Simplifying the last inequality gives us the result of the theorem.
\end{proof}
The established lower bound tells us that no scheduling policy can give us a better scaling result than exponential in $M$ (if the smallest and largest eigenvalues of $A$ and the smallest eigenvalue of $Q$ do not scale with $M$).

However, this lower bound may not be achievable in general, since there are a lot of relaxations for deriving the bound.

        \section{Simulations}\label{chap:simulations}
In this chapter, we provide numerical results comparing the long-term average normalized estimation error of different scheduling policies. For our simulations, we vary the number of sensors $M$ from $10$ to $150$. We consider both symmetric and asymmetric correlation structures. In subsection \ref{subsection:symm_cov_matrices}, all the sensors have identical variances and degrees of correlation, whereas in subsection \ref{subsection:asymm_cov_matrices}, we let the sensors have different variances and degrees of correlation. 
\subsection{Simulations with Discrete Wiener Process}
\subsubsection{Identical Variances}\label{subsection:symm_cov_matrices}
We first look at covariance matrices $Q$ where $q_{ii} = 1$, and $q_{ij} = \rho, i\neq j$. This is a covariance matrix with degree of correlation $\rho$.
\begin{figure}[H]
    \centering
    \includegraphics[width=1\linewidth]{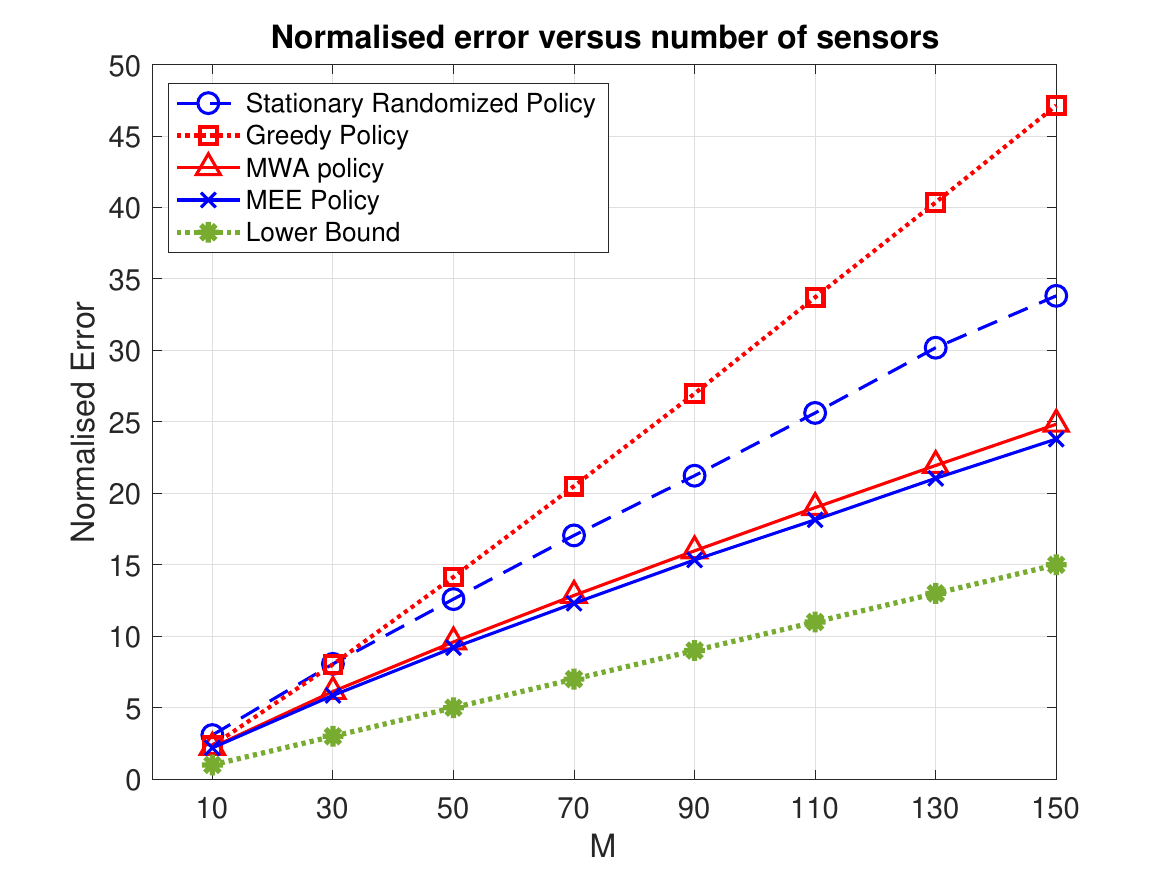}
    \caption{Comparison of long-term average expected error of scheduling policies for the correlation structure defined in section \ref{subsection:symm_cov_matrices} with $\rho = 0.8$}
    \label{fig:1rhorho}
\end{figure}
Figure \ref{fig:1rhorho} compares the empirical time average expected error of our proposed scheduling policies. We run the simulation for 20000 time steps and plot normalized error, i.e., $P_{\pi}/M$ for each scheduling policy $\pi$. We see that the MWA and MEE policies are close to the lower bound in \eqref{eqn:p_opt_lb_and_ub} and significantly outperform the greedy and the stationary randomized policy (as described in \cite[Section IV C]{kadota_scheduling_2018}) with process variances $q_{ii}$ as weights.

\subsubsection{Non-Identical Variances}\label{subsection:asymm_cov_matrices}
Next we consider covariance matrices $Q$ with the following correlation structure:
\begin{equation}
    \begin{aligned}
        Q
    &= \left[
    \begin{array}{ccc;{2pt/2pt}ccc} 
        1 &  \hdots & \rho  & 10\rho &  \hdots & 10\rho\\
        \vdots &  \ddots & \vdots & \vdots  & \ddots & \vdots \\
        \rho &  \hdots & 1  & 10\rho & \hdots & 10\rho\\
        \hdashline[2pt/2pt]
        10\rho &  \hdots & 10\rho & 100& \hdots & 100\rho\\
        \vdots &  \ddots & \vdots & \vdots  & \ddots & \vdots \\
        10\rho &  \hdots & 10\rho & 100\rho& \hdots & 100
    \end{array}
    \right]
    \end{aligned}
\end{equation}
where $M/2$ sensors have their variance as 1 and the other $M/2$ sensors have variance as 100.
More compactly, 
\begin{equation}
    Q = \begin{bmatrix}
        R & 10\rho\b{1}\b{1}^T\\
        10\rho\b{1}\b{1}^T & 100 R
    \end{bmatrix},
\end{equation}
where
    $R \in \R^{M/2 \times M/2}$ is a symmetric positive definite matrix, such that the diagonal elements $r_{ii} = 1$ and the off-diagonal elements $r_{ij} = \rho < 1$, and $\rho$ is the correlation coefficient independent of $M$, and
    $\b{1}$ denotes a $M/2$ length column vector of ones. 

Note that we get a family of ``well-behaved" symmetric positive definite correlation matrices $Q$ as we vary the parameter $M$. Further, the largest diagonal element of $Q$ is bounded by $100$, i.e., For any $M$, $\max_{i \in \{1,...,M\}} q_{ii} = 100$. The above correlation structure helps us study the performance of the proposed algorithms in the case of high correlation with vastly different variances. 
\begin{figure}[h]
    \centering
    \includegraphics[width=1\linewidth]{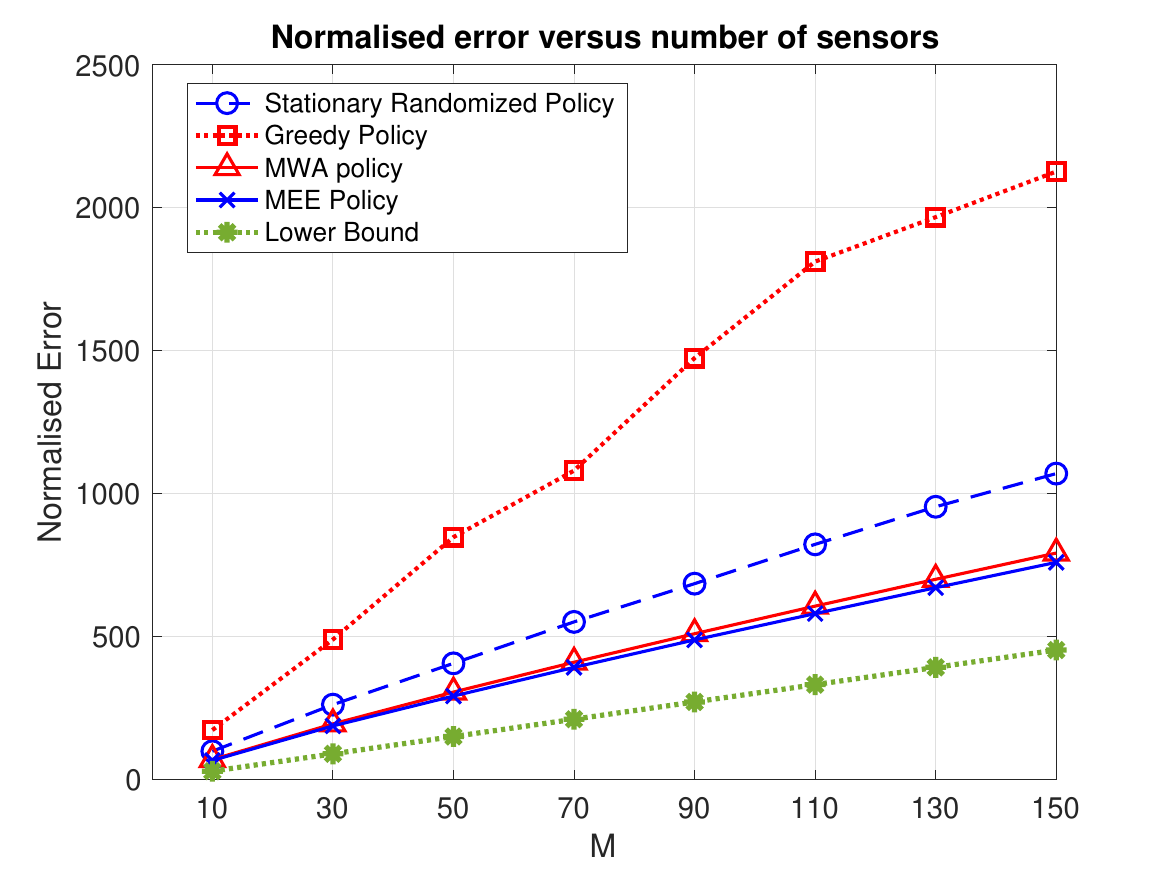}
    \caption{Comparison of long term average expected error of scheduling policies for the correlation structure defined in section \ref{subsection:asymm_cov_matrices} with $\rho = 0.8$}
    \label{fig:1rho_10rho}
\end{figure}
Figure \ref{fig:1rho_10rho} compares the empirical time average expected error of our proposed scheduling policies. We run the simulation for 20000 time steps and plot the normalized error, i.e., $P_{\pi}/M$ for each scheduling policy $\pi$.

From Figures \ref{fig:1rhorho} and \ref{fig:1rho_10rho}, we see that MEE and MWA significantly outperform the stationary randomized (SR) policy with process variances $q_{ii}$ as weights (see \cite{kadota_scheduling_2018} for more details on the SR policy). 
Secondly, we see that MWA and MEE are consistently better than the greedy scheduling policy, which involves greedily minimizing $\tr{P_t}$ at every time step. Moreover, the greedy policy displays inconsistent behavior. It performs even worse than the stationary randomized policy.

Finally, we see that the MWA policy and MEE policies have approximately the same empirical performance. This further strengthens our argument about AoI being a good heuristic for scheduling, even for correlated sensors. Nevertheless, we do observe that the MEE policy is marginally better than the MWA policy. This is because the  policy uses $p_{ii}(t)$ for scheduling, which exploits correlation in scheduling. However, the benefit of using correlation information in scheduling is marginal. 

\subsubsection{Performance vs Degree of Correlation}\label{subsection:row_changing}
Finally, we explore how the lower bound and upper bounds behave when the degree of correlation varies.  We look at covariance matrices $Q$ where $q_{ii} = 1$, and $q_{ij} = \rho, \forall i\neq j$. We fix the size of $Q$ to be $100 \times 100$ and vary $\rho$ from $0$ to $0.9$.

\begin{figure}[htbp]
    \centering
    \includegraphics[width=1\linewidth]{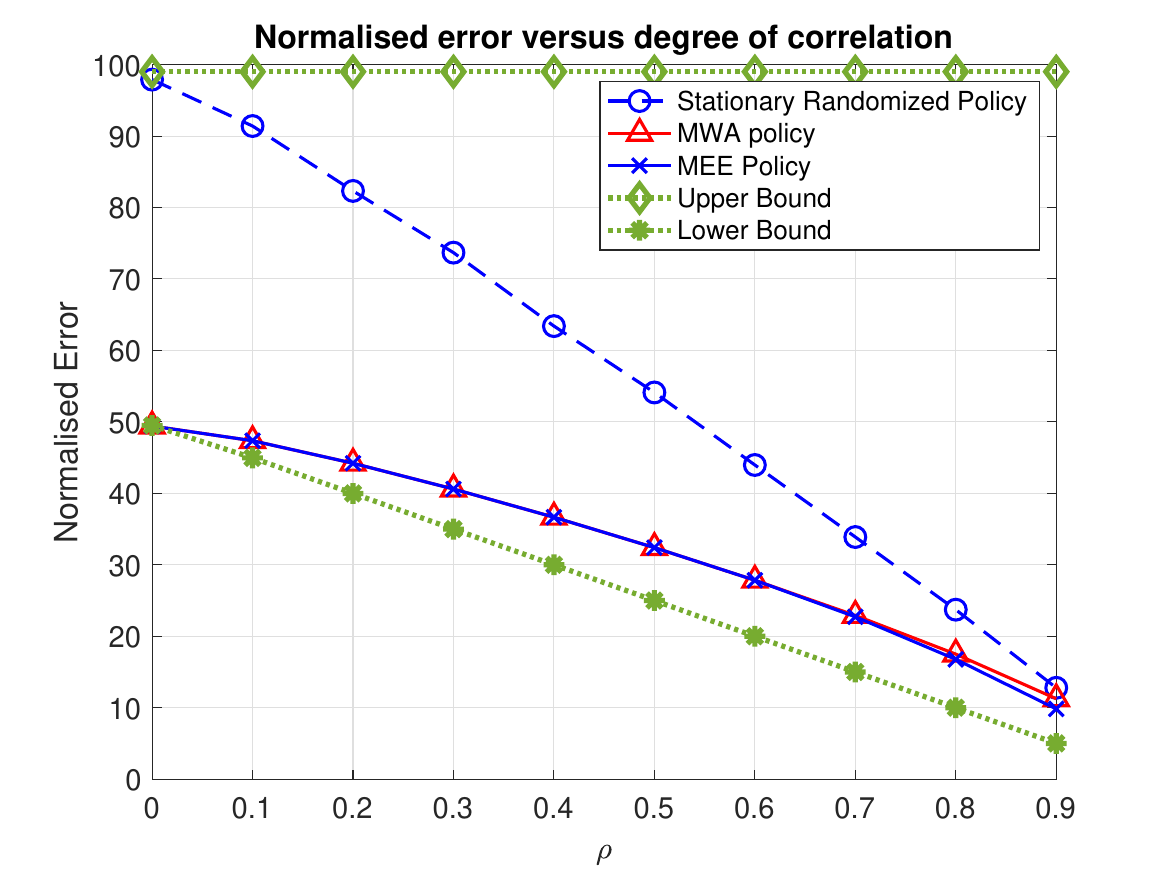}
    \caption{Comparison of long term average expected error of scheduling policies for the correlation structure defined in subsection \ref{subsection:row_changing}.}
    \label{fig:rho_comparison}
\end{figure}

The simulation for this comparison is plotted in Figure \ref{fig:rho_comparison}. In this figure, we see that the lower bound and upper bound are a factor of 2 apart, as expected from \cite{sun_age_2019}. As $\rho$ increases, the lower bound decreases but the upper bound does not change. This is because $q_{ii}$ is unaffected as $\rho$ increases, where as $\Tilde{q}_{ii}$ (as defined in Section \ref{chap:aoi-error-relation}) decreases. As $\rho$ gets closer to $1$, the lower bound approaches $0$. Also, we see that both MWA and MEE policies have identical performance, and are very close to the lower bound for $\rho = 0$. We see that MWA and MEE perform equally well for fairly low degrees of correlation. The separation in performance between MWA and MEE can only be observed at high degrees of correlation, i.e., $\rho = 0.8$ and $\rho=0.9$.

Lastly, we observe that the performance of the stationary randomized policy is very close to the upper bound for zero correlation. This is to be expected, as we derive the upper bound using the SR policy applied to a zero correlation setting. As $\rho$ increases, the error performance of the SR policy also improves because the Kalman estimator exploits the correlation to get better estimates.



\subsection{Simulations with diagonal $A$ matrices}
In this section, we compare the Whittle Index for Expected Error (WI-EE) policy from section \ref{sec:extension_schedule} with the lower bound developed in section \ref{sec:lb_extension}. In this section, we consider a system with the system matrix $A = 2I$. The covariance matrix is constructed in the same way as in subsection \ref{subsection:symm_cov_matrices}. 

\begin{figure}[H]
    \centering
    \includegraphics[scale=0.6]{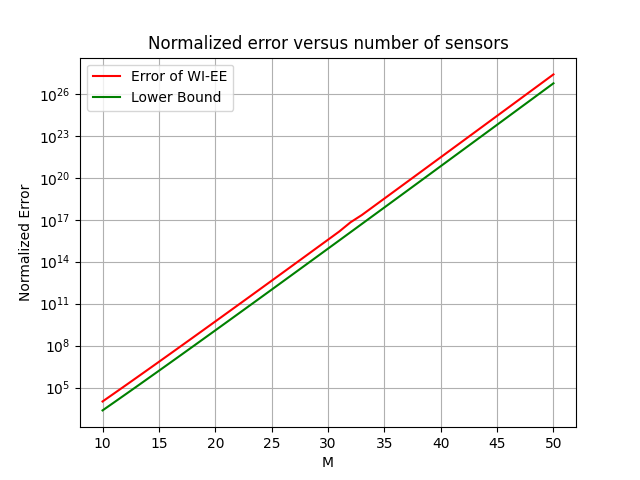}
    \caption{Comparison of the normalized error of WI-EE policy with the lower bound when $\rho = 0.8$.}
    \label{fig:WIEE_perf}
\end{figure}

Observe that the performance of the WI-EE is only a constant away from the lower bound in the log scale. This suggests that the WI-EE scales with $M$ in the same way as the lower bound does. The constant factor difference is $\sim 0.64$ in the log scale, which suggests that for this set of system settings, the WI-EE policy is roughly a factor of 4.4 away from the lower bound.
        \section{Conclusion}\label{chap:conc}
In this work, we successfully captured the importance of AoI as a metric for remote estimation and monitoring, even in the presence of correlation between sources in the network. From the bounds developed in Section \ref{chap:aoi-error-relation} and the empirical results in Section \ref{chap:simulations}, we show that the asymptotic estimation error does not benefit significantly from the presence of correlation. From the results in Section \ref{chap:scaling}, we establish that correlation will not lead to any order of magnitude improvement in the estimation error, and it is enough to look at AoI as a metric for scheduling policies. Finally, we also extend these results to more general classes of system matrices $A$.

        \bibliographystyle{IEEEtran}
        \bibliography{references}

        \appendices
        \section{Derivation of the Kalman Filter Recursion for the System Model}\label{appendix:kalman_derivation}
A general time-varying linear dynamic system driven by Gaussian noise is given as follows:
\begin{align}
    \b{x}_{t+1} &= F_t\b{x}_t + \b{w}_t \label{eqn:app_sys_dyn}\\ 
    \b{y}_t &= H_t\b{x}_t + \b{v}_t.\label{eqn:app_meas}
\end{align}
Here, $\b{x}_t \in \R^d$ is the state of the system and $\b{w}_t \in \R^d$ is the process innovation. $\b{w}_t$ is a multivariate normal random vector and $\b{w}_t \sim \mathcal{N}(0, Q_t)$. $F_t \in \R^{d \times d}$ is the matrix that captures the dynamics of the system at time $t$. $\b{y}_t \in \R^n$ is the measurement vector. It describes the measurement accessible to the estimator at time $t$. $\b{v}_t \in \R^n$ is the measurement noise and $\b{v}_t \sim \mathcal{N}(0, R_t)$. $H_t \in \R^{n \times d}$ is the measurement matrix. \\
The Kalman filtering algorithm gives us the optimal way to find the optimal MMSE estimate of the state $\b{x}_t$ using measurements, given the system and measurement matrices $F_t$ and $H_t$. The algorithm has two main steps 1) Prediction step and 2) Update step.\\
Define $\hat{\b{x}}_{t|t-1}$ be the state estimate given measurements only up to time $t-1$. This is the \textit{a priori}
 state estimate. Then, the \textit{a priori} covariance of this estimate is $P_{t|t-1} = \E\left[(\b{x}_t - \hat{\b{x}}_{t|t-1})(\b{x}_t - \hat{\b{x}}_{t|t-1})^T\right]$.\\
 After the measurement at time $t$, the Kalman filter algorithm gives us a method to compute a weighting/gain factor to update the state estimate. This is the Kalman gain $K_t$. Define the state estimate after measurement, i.e., the \textit{a posteriori} state estimate to be $\hat{\b{x}}_{t|t}$. Correspondingly, the covariance after measurement, or the \textit{a posteriori} covariance is $P_{t|t}= \E\left[(\b{x}_t - \hat{\b{x}}_{t|t})(\b{x}_t - \hat{\b{x}}_{t|t})^T\right]$.

\begin{enumerate}
    \item \textbf{Prediction Step}:
    \begin{itemize}
        \item Predict the \textit{a priori} state estimate:
        \begin{equation}
            \hat{\b{x}}_{t|t-1} = F_{t-1}\hat{\b{x}}_{t-1|t-1}.
        \end{equation}
        \item Predict the \textit{a priori} covariance:
        \begin{equation}
            P_{t|t-1} = F_{t-1}P_{t-1|t-1}F_{t-1}^T + Q_{t-1}.
        \end{equation}
    \end{itemize}
    \item \textbf{Update Step}:
    \begin{itemize}
        \item Compute the Kalman gain: 
        \begin{equation}
            K_{t} = P_{t|t-1}H_{t}^T(H_{t}P_{t|t-1}H_{t}^T + R_{t})^{-1}.
        \end{equation}
        \item Update the \textit{a posteriori} state estimate: 
        \begin{equation}
            \hat{\b{x}}_{t|t} = (I-K_tH_t)\hat{\b{x}}_{t|t-1} + K_t\b{y}_t.
        \end{equation}
        \item Update the \textit{a posteriori} covariance:
        \begin{equation}
            P_{t|t} = (I - K_tH_t)P_{t|t-1}.
        \end{equation}
    \end{itemize}
\end{enumerate}
\noindent
For more details, refer to \cite{kalman_new_1960, pei_elementary_2019}. 
\noindent
For the system model in \eqref{eqn:sys_model}, i.e.,
\begin{align*}
    \b{x}_{t+1} &= \b{x}_t + \b{w}_t\\
    y_t &= \b{c}_t^T\b{x}_t
\end{align*}
the Kalman filtering algorithm looks as follows:
\begin{enumerate}
    \item \textbf{Prediction Step}:
    \begin{itemize}
        \item Predict the \textit{a priori} state estimate:
        \begin{equation}
            \hat{\b{x}}_{t|t-1} = \hat{\b{x}}_{t-1|t-1}.
        \end{equation}
        \item Predict the \textit{a priori} covariance:
        \begin{equation}
            P_{t|t-1} = P_{t-1|t-1} + Q.
        \end{equation}
    \end{itemize}
    \item \textbf{Update Step}:
    \begin{itemize}
        \item Compute the Kalman gain: 
        \begin{equation}
            K_{t} = \frac{P_{t|t-1}\b{c}_{t}}{\b{c}_t^TP_{t|t-1}\b{c}_t}.
        \end{equation}
        \item Update the \textit{a posteriori} state estimate: 
        \begin{equation}
            \hat{\b{x}}_{t|t} = (I-K_t\b{c}_t^T)\hat{\b{x}}_{t|t-1} + K_ty_t.
        \end{equation}
        \item Update the \textit{a posteriori} covariance:
        \begin{equation}
            P_{t|t} = (I - K_t\b{c}_t^T)P_{t|t-1}.
        \end{equation}
    \end{itemize}
\end{enumerate}
In the main work, we avoid the double subscript, and when we refer to the ``state estimate", we refer to the \textit{a posteriori} state estimate. Therefore, $\hat{\b{x}}_t$ refers to $\hat{\b{x}}_{t|t}$. Notice that the recursion for the \textit{a priori} covariance matrix is more elegant than that for the \textit{a posteriori} covariance.
The recursion for $P_{t|t-1}$ is given as 
\begin{align}
    P_{t+1|t} &= P_{t|t} + Q \nonumber\\
    &= P_{t|t-1} - K_t\b{c}_t^TP_{t|t-1} + Q \nonumber\\
    &=  P_{t|t-1} - \frac{P_{t|t-1}\b{c}_t\b{c}_t^TP_{t|t-1}}{\b{c}_t^T  P_{t|t-1}\b{c}_t} + Q.
\end{align}
The matrix $P_t$ in \ref{lemma:Kalman_for_sys_model} is the \textit{a priori} covariance matrix $P_{t|t-1}$ described above. And hence, it follows the recursion,
\begin{align}
    P_{t+1} = P_{t} - \frac{P_{t}\b{c}_t\b{c}_t^TP_{t}}{\b{c}_t^T  P_{t}\b{c}_t} + Q.
\end{align}
Since, $P_{t+1|t} = P_{t|t} + Q$, we also have the relation, 
\begin{align}
    \E\left[(\b{x}_t - \hat{\b{x}}_t)(\b{x}_t - \hat{\b{x}}_t)^T\right] = P_{t+1} - Q.
\end{align}
Consequently, comparing the matrices elementwise, we get,
\begin{align}
    \E\left[(x^i_t - \hat{x}^i_t)^2\right] = (P_{t+1} - Q)_{ii}.
\end{align}

\section{Proof of Theorem \ref{thm:upper_bound} }\label{appendix:thm_upper_bound_proof}
Define $u_i(t)$ as the Boolean variable which tells us if $i$ has been scheduled, i.e., $u_i(t) = 1$, if $i$ is scheduled in timeslot $t$, and $u_i(t) = 0$ otherwise. 
We can rewrite the matrix equation \eqref{eqn:Kalman_cov_eq}, i.e. the Kalman recursion involving $P_t$ at time $t$, for each element of the matrix $P_t$ as follows
\begin{equation}
\label{eqn:Kalman_pt_elementwise}
    p_{ij}(t+1) = p_{ij}(t) - \sum_{l=1}^M \frac{p_{il}(t)p_{jl}(t)}{p_{ll}(t)}u_l(t) + q_{ij}
\end{equation}
Now, let's consider \eqref{eqn:Kalman_pt_elementwise} but for the special case of diagonal elements of $P_t$.
\begin{align}
    p_{ii}(t+1) &= p_{ii}(t) - \sum_{l=1}^M \frac{p_{il}(t)^2}{p_{ll}(t)}u_l(t) + q_{ii} \label{eqn:diag_elem}
\end{align}
By definition, since $P_t$ and $Q$ are covariance matrices, they must be positive semi-definite, i.e. $P_t \succeq 0$ and $Q \succeq 0$. This, in turn, implies that $p_{ii}(t) \geq 0$ and $q_{ii} \geq 0$ for all $i$. So, note that $p^2_{il}(t)/p_{ll}(t) \geq 0$. Hence, 
\begin{equation}
    p_{ii}(t+1) \leq p_{ii}(t) + q_{ii}.
\end{equation}
and so,
\begin{equation}\label{eqn:telescopic_thm_1}
    p_{ii}(t+1) - p_{ii}(t) \leq q_{ii}.
\end{equation}

This is a recursive relationship. 
Note that if sensor $i$ is scheduled at time $t_i$, i.e. $\b{c}_{t_i} = \b{e}_i$, then $p_{ij}(t_i+1) = p_{ji}(t_i+1) = q_{ij}$ and $p_{ii}(t_i+1) = q_{ii}$. In other words, if we observe sensor $i$ at time $t_i$, then it's expected error in the following time-slot goes down to the variance in the $i$th dimension of the one-step noise increment. Further, the covariance between $i$ and any other source $j$ reduces to the covariance between the $i$th and $j$th dimension of the one-step noise increment.

Therefore, summing \eqref{eqn:telescopic_thm_1} until the last time  $i$ was scheduled, we get
\begin{align}
    p_{ii}(t+1) - p_{ii}(t_i+1) \leq (t - t_i)\cdot q_{ii}
\end{align}
But, $p_{ii}(t_i + 1) = q_{ii}$. Also note that $t-t_i$ is essentially the age of sensor $i$, $h_i(t)$. Consequently,
\begin{align}\label{eqn:pii_ub}
    p_{ii}(t+1) \leq (h_i(t)+1)\cdot q_{ii}
\end{align}
and,
\begin{align}
   \E[(x^i_t - \hat{x}^i_t)^2] \leq q_{ii}\E[h_i(t)].
\end{align}

\section{Proof of Theorem \ref{thm:lower_bound} }\label{appendix:thm_lower_bound_proof}
To prove theorem \ref{thm:lower_bound}, we need to establish certain lemmas. 
\begin{lemma}\label{lemma:pt_geq_q}
If $P_0 \succ 0, Q \succ 0$, then $P_t \succeq Q$, $\forall t$. 
\end{lemma}
\begin{proof}
Assume that $P_{t} \succeq Q \succ 0$. From (\ref{eqn:kalman_prob_prior}), we have
\begin{align*}
    P_{t+1} &= P_t - \frac{P_t\b{c}_t\b{c}_t^TP_t}{\b{c}_t^TP_t\b{c}_t} + Q\\
   \implies  P_{t+1} - Q &= P_t\left(I - \frac{\b{c}_t\b{c}_t^TP_t}{\b{c}_t^TP_t\b{c}_t}\right).
\end{align*}
Consider matrix $I - \frac{\b{c}_t\b{c}_t^TP_t}{\b{c}_t^TP_t\b{c}_t}$.
\begin{align*}
    \left(I - \frac{\b{c}_t\b{c}_t^TP_t}{\b{c}_t^TP_t\b{c}_t}\right)\b{c}_t = \b{c}_t - \frac{\b{c}_t(\b{c}_t^TP_t\b{c}_t)}{\b{c}_t^TP_t\b{c}_t} = \b{c}_t - \b{c}_t = 0.
\end{align*}
So, $0$ is an eigenvalue of $I - \frac{\b{c}_t\b{c}_t^TP_t}{\b{c}_t^TP_t\b{c}_t}$. Since $\b{c}_t$ is a vector in $R^{M}$, there exist $M-1$ orthogonal vectors to $\b{c}_t$ which form an orthogonal basis with $\b{c}_t$ for $R^{M}$. Let $\mathcal{V}$ denote the basis of these vectors such that $\b{v} \in \mathcal{V} \implies \b{c}_t^T\b{v} = 0$ and $\b{v}_1, \b{v}_2 \in \mathcal{V} \implies \b{v}_1^T\b{v}_2 = 0$. Since, $P_t \succ 0$, we have that $P_t^{-1}$ exists. Now, for any $\b{v} \in \mathcal{V}$,
\begin{align*}
    \left(I - \frac{\b{c}_t\b{c}_t^TP_t}{\b{c}_t^TP_t\b{c}_t}\right)(P_t^{-1}v) &= P_t^{-1}\b{v} -  \frac{\b{c}_t\b{c}_t^TP_tP_t^{-1}\b{v}}{\b{c}_t^TP_t\b{c}_t}\\ &=  P_t^{-1}\b{v} -  \frac{\b{c}_t\b{c}_t^T\b{v}}{\b{c}_t^TP_t\b{c}_t}\\
    &= P_t^{-1}\b{v}
\end{align*}
Since this is true for all $M-1$ unique vectors in $V$, we infer that the eigenspace of eigenvalue $1$ of $I -  \frac{\b{c}_t\b{c}_t^TP_t}{\b{c}_t^TP_t\b{c}_t}$ is $\R^{M-1}$. We have now found all $M$ eigenvalues of $I -  \frac{\b{c}_t\b{c}_t^TP_t}{\b{c}_t^TP_t\b{c}_t}$. The maximum eigenvalue is $1$ and the minimum eigenvalue is $0$. From this, we have $I -  \frac{\b{c}_t\b{c}_t^TP_t}{\b{c}_t^TP_t\b{c}_t} \succeq 0$. Therefore, $P_t\left(I -  \frac{\b{c}_t\b{c}_t^TP_t}{\b{c}_t^TP_t\b{c}_t}\right) \succeq 0$. Hence, $P_{t+1} - Q \succeq 0 \implies P_{t+1} \succeq Q$.\\
Now, note that if $P_0 \succ 0$, then $P_1 = P_0 + P_0\b{c}_0\b{c}_0^TP_0/\b{c}_0^TP_0\b{c}_0 + Q \succeq Q$ from the same analysis as above. This proves the induction hypothesis.
\end{proof}
\begin{lemma}\label{lemma:lb_main_lem} For $A \succeq B \succ 0$, $A, B \in \R^{n\times n}$, and any vector $\b{p} \in \R^n$, we have
    \begin{align*}
        \max_{\b{b} \in \R^{M-1}} \frac{\b{b}^T\b{p}\b{p}^T\b{b}}{\b{b}^TA\b{b}} &= \b{p}^TA^{-1}\b{p} \leq \b{p}^TB^{-1}\b{p}
    \end{align*}
\end{lemma}
\begin{proof}
    Since $A$ is positive definite, $A$ is invertible. We can also factor it as $A = D^TD$, where $D$ is also invertible. Using this,
    \begin{align}
        \frac{\b{b}^T\b{p}\b{p}^T\b{b}}{\b{b}^TA\b{b}} = \frac{\b{b}^T\b{p}\b{p}^T\b{b}}{\b{b}^TD^TD\b{b}}  
    \end{align}
    Let $\b{b}' = D\b{b}$. $\b{b}'$ can be any vector in $\R^{M-1}$ because $D$ is invertible. Hence,
    \begin{align}
        \frac{\b{b}^T\b{p}\b{p}^T\b{b}}{\b{b}^TA\b{b}} = \frac{\b{b}'^TD^{-T}\b{p}\b{p}^TD^{-1}\b{b}'}{\b{b}'^T\b{b}'},
    \end{align}
and,
    \begin{align}
        \max_{\b{b} \in \R^n}\frac{\b{b}^T\b{p}\b{p}^T\b{b}}{\b{b}^TA\b{b}} &= \max_{\b{b}' \in \R^n}\frac{\b{b}'^TD^{-T}\b{p}\b{p}^TD^{-1}\b{b}'}{\b{b}'^T\b{b}'}\nonumber\\
        &= \lambda_1(D^{-T}\b{p}\b{p}^TD^{-1}).  
    \end{align}
    Here, $\lambda_1(.)$ denotes the largest eigenvalue of a matrix. Note that $D^{-T}\b{p}\b{p}^TD^{-1}$ is a rank one matrix, as it is the outer product of two vectors in $\R^{M-1}$. Hence, the matrix has at least $M-2$ zero eigenvalues, and the largest (the only positive) eigenvalue is $\b{p}^TD^{-1}D^{-T}\b{p}$. But,
    \begin{equation}
       \b{p}^TD^{-1}D^{-T}\b{p} = \b{p}^T(D^TD)^{-1}\b{p}= \b{p}^TA^{-1}\b{p}.
    \end{equation}
    Since $A \succeq B \succ 0$, we have $B^{-1} \succeq A^{-1}$ (Lemma \ref{lemma:inv_psd} in Appendix \ref{appendix:inv_psd}). And hence,
     \begin{align}
        \max_{\b{b} \in \R^n}\frac{\b{b}^T\b{p}\b{p}^T\b{b}}{\b{b}^TA\b{b}}= \b{p}^TA^{-1}\b{p} 
        \leq \b{p}^TB^{-1}\b{p}.
    \end{align}
    
\end{proof}
With the help of the above lemmas, we are ready to prove Theorem \ref{thm:lower_bound}.
\begin{proof}[Proof of Theorem \ref{thm:lower_bound}]
For any policy $\pi$, define $\mathcal{T}^i_{\pi} = \{t \in \N | u_{i}(t) = 1\}$.
For a lower bound, let us relax the problem to include all linear combinations of other nodes when $i$ is not scheduled, i.e., $\b{c}_t \in \R^{M-1}$ and $c_{t,i} = 0$, $\forall t \notin \mathcal{T}^i_{\pi}$, and $\forall t \in \mathcal{T}^i_{\pi}$, $c^{\pi}_t = \b{e}_i$. For notational ease, assume $i=1$ (the following analysis will hold for all $i$). With this relaxation, we proceed to obtain a lower bound to the estimation error of $x^1_t$ in this setting.

To understand the covariance evolution, let us split up $P_t$ and $Q$ into blocks as follows.
\begin{equation}\label{eqn:P_partition}
    P_t = \begin{bmatrix}
    p_{11}(t) & \b{p}^T_1(t)\\
    \b{p}_1(t) & R_t
    \end{bmatrix}
\end{equation}

\begin{equation}\label{eqn:Q_partition}
    Q = \begin{bmatrix}
    q_{11} & \b{q}^T_1\\
    \b{q}_1 & Q_{-1}
    \end{bmatrix}
\end{equation}
Let us understand what these blocks represent. $p_{11}(t)$ represents the variance of estimation error $x^1_t - \hat{x}^1_t$. $\b{p}_1(t)$ represents the covariance of $(x^1_t - \hat{x}^1_t)$ with the other error processes $(x^j_t - \hat{x}^j_t)$. $R_t$ represents the covariance matrix of the remaining $M-1$ sensors.\\
Similarly, $q_{11}$ represents the noise variance of $w^1_t$. $\b{q}_1$ represents the covariance of $w^1_t$ with the other noise terms  $w^j_t$. $Q_{-1}$ represents the covariance matrix of the remaining $M-1$ noise terms.

We split up the matrix into the above blocks to understand the best error variance we can achieve for sensor 1, without scheduling sensor 1. Let $t'_1$ denote the last time sensor 1 was scheduled, and $t_1$ denote the last time sensor 1 was scheduled prior to $t'_1$. For any $t'_1 > t > t_1$, we have $\b{c}_t = [0 \; \b{b}_t^T]^T$. And so, we can rewrite the matrix evolution 
\begin{equation*}
    P_{t+1} = P_t - \frac{P_t\b{c}_t\b{c}_t^TP_t}{\b{c}_t^TP_t\b{c}_t} + Q
\end{equation*}
in terms of its respective blocks,
\begin{equation}\label{eqn:diag_evol}
        p_{11}(t+1) = p_{11}(t) - \frac{\b{b}_t^T(\b{p}_1(t)\b{p}_1(t)^T)\b{b}_t}{\b{b}_t^TR_t\b{b}_t} + q_{11},
\end{equation}
\begin{equation}\label{eqn:pvec_evol}
        \b{p}_{1}(t+1) = \b{p}_{1}(t) - \frac{R_t\b{b}_t\b{b}_t^T\b{p}_1(t)}{\b{b}_t^TR_t\b{b}_t} + \b{q}_1,
\end{equation}
\begin{equation}
        R_{t+1} = R_t - \frac{R_t\b{b}_t\b{b}_t^TR_t}{\b{b}_t^TR_t\b{b}_t} + Q_{-1}.
\end{equation}
Note that the evolution of $R_t$ is very similar to the evolution of $P_t$, and hence, from the same analysis as in Lemma \ref{lemma:pt_geq_q}, we conclude that $R_t \succeq Q_{-1} \succ 0$. We now focus on (\ref{eqn:diag_evol}), and lower bound it. Clearly, 
\begin{equation*}
    p_{11}(t+1) \geq p_{11}(t) + q_{11} - \max_{\b{b} \in \R^{M-1}} \frac{\b{b}^T\b{p}_1(t)\b{p}_1^T(t) \b{b}}{\b{b}^TR_t\b{b}} 
\end{equation*}

Using Lemma \ref{lemma:lb_main_lem}, we get a lower bound on the one step increment in estimation error, i.e.,
\begin{equation}
     p_{11}(t+1) \geq p_{11}(t) + q_{11} - \b{p}_1^T(t)Q_{-1}^{-1}\b{p}_1(t)
\end{equation}
\begin{equation}
         \implies  p_{11}(t+1) - p_{11}(t)\geq q_{11} - \b{p}_1^T(t)Q_{-1}^{-1}\b{p}_1(t).
\end{equation}
Using the telescopic nature of the left-hand side, we have,
\begin{align*}
     p_{11}(t+1) &- p_{11}(t_1+1)\\ 
     &\geq (t - t_1) q_{11} - \sum_{k=t_1+1}^{t} \b{p}_1^T(k)Q_{-1}^{-1}\b{p}_1(k).
\end{align*}
Since $t_1$ was a time at which sensor 1 was scheduled, $p_{11}(t_1+1) = q_{ii}$ and hence, for $t'_1 > t > t_1$,
\begin{equation}\label{eqn:lb_almost}
    p_{11}(t+1) \geq (t - t_1 + 1) q_{11} - \sum_{k=t_1+1}^{t} \b{p}_1^T(k)Q_{-1}^{-1}\b{p}_1(k).
\end{equation}
Note that left-hand side of (\ref{eqn:lb_almost}) is independent of the choices of $\b{b}_k$ since time $t_1$, whereas the right-hand side depends on the choices of $\b{b}_k$ as seen in the summation term. Also, the evolution of $\b{p}_1(t)$ and $R_t$ is unaffected by the values of $p_{11}(t)$. This tells us that the bound is true regardless of the choices of $\b{b}_k$ made since time $t_1$. We now focus on the evolution of $\b{p}_{k}$. From (\ref{eqn:pvec_evol}), we have
\begin{equation}
    \b{p}_{1}(t+1) = \b{p}_{1}(t) - \frac{R_t\b{b}_t\b{b}_t^T\b{p}_1(t)}{\b{b}_t^TR_t\b{b}_t} + \b{q}_1.
\end{equation}
 Since the LHS of the bound in (\ref{eqn:lb_almost}) is independent of the choices of $b_k$, choose $\b{b}_t = R_t^{-1}\b{p}(t)$. On substituting $b_{t} = R^{-1}_{t}\b{p}(t)$, we have 
\begin{align}
    \b{p}_1(t+1) &= \b{p}_1(t) + \b{q}_1 \nonumber \\
    & \;\;\; - \frac{R_t(R_t^{-1}\b{p}_1(t))(\b{p}^T_1(t)R_t^{-1})\b{p}_1(t)}{\b{p}_1^T(t)R_t^{-1}R_tR_t^{-1}\b{p}_1^T(t)} \nonumber \\ 
    &= \b{p}_1(t) - \b{p}_1(t) + \b{q}_1\nonumber \\ 
    &= \b{q}_1.
\end{align}
At time $t_1$, we know that $p_{1j}(t) = q_{1j}$, $\forall j$. Hence, $\b{p}(t_1) = \b{q}_1$. On using the above mentioned choices of $\b{b}_k$ in (\ref{eqn:lb_almost}), we have
\begin{equation}
    p_{11}(t+1) \geq  (t - t_1 + 1) q_{11} - (t-t_1) \b{q}_1^TQ_{-1}^{-1}\b{q}_1.
\end{equation}
But note that $h_1(t) = t-t_1$ is the age of sensor 1. Hence, we can generalize the lower bound to all $t$ by using the age process.
\begin{equation*}
    p_{11}(t+1) \geq  q_{11} + h_1(t)(q_{11} -  \b{q}_1^TQ_{-1}^{-1}\b{q}_1).
\end{equation*}

Note that all our above analysis does not depend on choosing to analyze with sensor 1. The analysis holds for all sensors, and hence,
\begin{equation}
    p_{ii}(t+1) \geq  q_{ii} + q'_{ii}h_i(t),
\end{equation}
where we define $q'_{ii} = q_{ii} - \b{q}_i^TQ_{-i}^{-1}\b{q}_i$.
Since $Q \succ 0$, we have $q_{ii} > 0$ and $q'_{ii} = q_{ii} - \b{q}_i^TQ_{-i}^{-1}\b{q}_i > 0$ from the properties of positive definite matrices and their Schur complements \cite{boyd_block_nodate}. Consequently, we get
\begin{equation}
    \E[(x^i_t - \hat{x}^i_t)^2] = p_{ii}(t+1) - q_{ii} \geq q'_{ii}h_i(t).
\end{equation}

\end{proof}

\section{Proof of Theorem \ref{thm:ext_A_upper_bound} }\label{appendix:extension_upper_bound}
The proof for the upper bound on every timestep for the case of diagonal $A$ matrices is very similar to the discrete Wiener process. Writing the elementwise evolution of the diagonal \textit{apriori} covariance matrix. Note that for the case with $A$ matrices, equation \eqref{eqn:Kalman_cov_eq} gets modified to
\begin{equation}
    P_{t+1} = AP_tA^T - \frac{AP_t\b{c}_t\b{c}_t^TP_tA^T}{\b{c}_t^TP_t\b{c}_t} + Q.
\end{equation}
So, the diagonal element update becomes 
\begin{align}
    p_{ii}(t+1) &= a_{ii}^2 p_{ii}(t) - \sum_{l=1}^M \frac{a_{ii}a_{ll}p_{il}(t)^2}{p_{ll}(t)}u_l(t) + q_{ii} \label{eqn:diag_elem_A}
\end{align}
Hence, 
\begin{equation}
    p_{ii}(t+1) \leq a_{ii}^2p_{ii}(t) + q_{ii}.
\end{equation}
Therefore, summing until the last time  $i$ was scheduled, we get
\begin{align}
    p_{ii}(t+1) - p_{ii}(t_i+1) \leq \sum_{k=t_i}^t a_{ii}^{2(k-t_i)}\cdot q_{ii}
\end{align}
But, $p_{ii}(t_i) = 0$. Also note that $t-t_i$ is essentially the age of sensor $i$, $h_i(t)$. Consequently,
\begin{align}
   \E[(x^i_t - \hat{x}^i_t)^2] \leq q_{ii}\E\left[\sum_{k=0}^{h_i(t)}a_{ii}^{2k}\right].
\end{align}

\section{Proof of Theorem \ref{thm:ext_A_lower_bound} }\label{appendix:extension_lower_bound}
The proof for the lower bound on every timestep for the case of diagonal $A$ matrices is very similar to the discrete Wiener process. 
We consider the same partitions of the $P_t$ matrix as defined in \eqref{eqn:P_partition} and \eqref{eqn:Q_partition}, along with $\b{c}_t = [0 \;\;\b{b}_t]^T$. Without loss in generality, consider only $i=1$ (as in Appendix \ref{appendix:thm_lower_bound_proof}). Then, the respective evolution of the partitions is 
\begin{equation}
    p_{11}(t+1) = a^2_{11}p_{11}(t) - a_{11}^2\frac{\b{b}^T\b{p_1}(t)\b{p_1}^T(t)\b{b}}{\b{b}^TR_t\b{b}} + q_{11} 
\end{equation}
\begin{equation}
    \b{p}_1(t+1) = a_{11}\begin{bmatrix}
        a_{22} & & \\
        & \ddots & \\
        & & a_{MM}
    \end{bmatrix}\left(I - \frac{R_t\b{b}\b{b}^T}{\b{b}^TR_t\b{b}}\right)\b{p_1}(t) + \b{q}_{1} 
\end{equation}
And,
\begin{align}
    &R_{t+1}\nonumber\\ 
    &= \begin{bmatrix}
        a_{22} & & \\
        & \ddots & \\
        & & a_{MM}
    \end{bmatrix}\left(R_t - \frac{R_t\b{b}\b{b}^TR_t}{\b{b}^TR_t\b{b}}\right) \begin{bmatrix}
        a_{22} & & \\
        & \ddots & \\
        & & a_{MM}
    \end{bmatrix}\nonumber\\
    &\qquad+ Q_{-1}.
\end{align}
From Lemma \ref{lemma:lb_main_lem}, 
\begin{align}\label{eqn:ind_step}
    p_{11}(t+1) &\geq a^2_{11}p_{11}(t) - a_{11}^2\b{p}^T_1(t)R^{-1}_1\b{p}_1(t) + q_{11} \nonumber\\
    &\geq a^2_{11}p_{11}(t) - a_{11}^2\b{p}^T_1(t)Q^{-1}_{=1}\b{p}_1(t) + q_{11}.
\end{align}
Note that we can apply \eqref{eqn:ind_step} inductively till the last time sensor 1 was scheduled. Let $t_1$ be the last time sensor 1 was scheduled. Then,
\begin{align}
    p_{11}(t) &\geq a_{ii}^{2(h_{i}(t) + 1)}p_{11}(t_1)\nonumber\\
    &\qquad - \sum_{k=1}^{h_{ii}(t)+1}a_{11}^{2k} \b{p_1}^T(k)Q_{-1}^{-1}\b{p_1}(k) + q_{11}\sum_{k=0}^{h_{ii}(t)}a_{11}^{2k}.
\end{align}
This lower bound is true for all trajectories of $\b{p}_1(k)$ as the trajectory of $\b{p}_{1}(k)$ does not depend on the evolution of $p_{11}(k)$. Hence, if we choose a $\B{b}$ every step such that $\b{p}_{1}(k) = \b{q}_1$ at every step, we get
\begin{align}
    p_{11}(t) \geq q_{11}\sum_{k=0}^{h_{ii}(t)+1}a_{11}^{2k} - \b{q}_1^TQ_{-1}^{-1}\b{q}_1\sum_{k=1}^{h_{ii}(t)+1}a_{11}^{2k}.
\end{align}
Note that,
\begin{align}
    \E[(x^i_t - \hat{x}^i_t)^2]  = \frac{p_{ii}(t) - q_{ii}}{a_{ii}^2} \geq \sum_{i=1}^M\left(q_{ii}-\b{q}_i^TQ_{-1}^{-i}\b{q}_i\right)\sum_{k=0}^{h_{ii}(t)}a_{11}^{2k}.
\end{align}
Finally, on summing $\sum_{i=1}^M\E[(x^i_t - \hat{x}^i_t)^2]$, we get the desired result.

\section{Proof of Theorem \ref{thm:mee_policy}}\label{appendix:mee_proof}
\begin{proof}
For all $\pi \in \Pi$, we have 
\begin{align}
     \E[\Delta_{\pi^{MEE}}(t) | P_t] &\leq \sum_{i=1}^M\sqrt{q_{ii}} - \sum_{i=1}^M\frac{p_{ii}(t)}{\sqrt{q_{ii}}}\E\left[u^{\pi^{MEE}}_i(t) \big| P_t\right] \nonumber \\
     &\leq \sum_{i=1}^M\sqrt{q_{ii}} - \sum_{i=1}^M\frac{p_{ii}(t)}{\sqrt{q_{ii}}}\E\left[u^\pi_i(t) \big| P_t\right]. \label{eqn:drift_mee_ineq}
\end{align}
This is because the MEE policy minimizes the upper bound of the Lyapunov drift in \eqref{eqn:mee_drift_ub}.
Let $\pi^{SR}$ be the optimal stationary random policy for the expected weighted sum of age of information (EWSAoI) with $q_{ii}$ as the weights (refer to \cite[Section IV C]{kadota_scheduling_2018}).
Let us choose $\pi^{SR}$ in the RHS of inequality (\ref{eqn:drift_mee_ineq}). Recall that the optimal stationary random policy schedules sensor $i$ with probability $\sqrt{q_{ii}}/\sum_{j=1}^M \sqrt{q_{jj}}$\cite[Section IV C]{kadota_scheduling_2018}. Hence,
\begin{align}
    \E[\Delta_{\pi^{MEE}}(t)|P_t] \leq \sum_{i=1}^M\sqrt{q_{ii}} - \sum_{i=1}^M\frac{p_{ii}(t)}{\sqrt{q_{ii}}}\frac{\sqrt{q_{ii}}}{\sum_{j=1}^M\sqrt{q_{jj}}}.
\end{align}
We take the expectation on both sides and employ the law of iterated expectation to obtain,
\begin{align}
    \E[\Delta_{\pi^{MEE}}(t)] \leq \sum_{i=1}^M\sqrt{q_{ii}} - \frac{\sum_{i=1}^M\E[p_{ii}(t)]}{\sum_{i=1}^M\sqrt{q_{ii}}}.
\end{align}
From the telescopic nature of $\Delta_{\pi^{MEE}}(t)$, we obtain the following inequality on summing the above expression from $t=1$ to $t=T$.
\begin{align}
    \E[L(T+1)] - \E[L(1)] \leq T\sum_{i=1}^M\sqrt{q_{ii}} - \frac{\sum_{t=1}^{T}\sum_{i=1}^M\E[p_{ii}(t)]}{\sum_{i=1}^M\sqrt{q_{ii}}}.
\end{align}
On rearrangement of terms, we end up with
\begin{align}
    \frac{1}{T}\sum_{t=1}^{T}\sum_{i=1}^M\E[p_{ii}(t)] \leq \left(\sum_{i=1}^M\sqrt{q_{ii}}\right)^2+\left(\sum_{i=1}^{M}\sqrt{q_{ii}}\right)\frac{\E[L(1)]}{T},
\end{align}
where the last inequality is true because $L(T+1)$ is non-negative. Note that $L(1)$ is a finite quantity, and hence, as $T \to \infty$, we get,
\begin{equation}
     \lim_{T \to \infty}\frac{1}{T}\sum_{t=1}^T\sum_{i=1}^M\E[p_{ii}(t)] \leq \left(\sum_{i=1}^M\sqrt{q_{ii}}\right)^2.
\end{equation}
Therefore,
\begin{align}
    P_{MEE} &= \lim_{T \to \infty}\frac{1}{T}\sum_{t=0}^{T-1}\E\left[||\b{x}_t - \hat{\b{x}}_t||^2_2\right]\nonumber\\
    &=\frac{1}{T}\sum_{t=1}^T\sum_{i=1}^M\E[p_{ii}(t)] - \tr{Q}\nonumber\\
    &\leq \left(\sum_{i=1}^M\sqrt{q_{ii}}\right)^2 - \sum_{i=1}^M q_{ii}.
\end{align}
It also turns out that using the age-based lower bound developed in Section \ref{chap:aoi-error-relation}, we can develop a lower bound for $P_{OPT}$. This is discussed in detail in Section \ref{chap:scaling}. The result arises from the result of \cite[Theorem 6]{kadota_scheduling_2018}. However, we just state and use the result here.
The optimal estimation error $P_{OPT}$ is lower bounded as
\begin{equation}
    P_{OPT} \geq \frac{1}{2}\left(\sum_{i=1}^M \sqrt{q'_{ii}}\right)^2 - \frac{1}{2}\sum_{i=1}^Mq'_{ii}.
\end{equation}
Therefore, the ratio, $P_{MEE}/P_{OPT}$ can be upper bounded by
\begin{equation}
    \frac{P_{MEE}}{P_{OPT}} \leq 2\frac{\left(\sum_{i=1}^M\sqrt{q_{ii}}\right)^2 - \sum_{i=1}^M q_{ii}}{ \left(\sum_{i=1}^M\sqrt{q'_{ii}}\right)^2 - \sum_{i=1}^M q'_{ii}}.
\end{equation}
\end{proof}

\section{Proof of Theorem \ref{thm:mwa_policy}}\label{appendix:mwa_proof}
\begin{proof}
    For all $\pi \in \Pi$,
\begin{align}\label{eqn:drift_mwa_inequality}
 \E&[\Delta_{\pi^{MWA}}(t) | \b{h}_t]\nonumber\\ 
 &= \sum_{i=1}^M\sqrt{q_{ii}} - \sum_{i=1}^M\sqrt{q_{ii}}(h_{i}(t-1)+1)\E\left[u^{\pi^{MWA}}_i(t)\bigg| \b{h}_t\right]\nonumber\\
 &\leq \sum_{i=1}^M\sqrt{q_{ii}} - \sum_{i=1}^M\sqrt{q_{ii}}(h_i(t-1)+1)\E\left[u^{\pi}_i(t)\bigg| \b{h}_t\right].
\end{align}
This is because the MWA policy minimizes the Lyapunov drift in \eqref{eqn:mwa_drift}.
Let $\pi^{SR}$ be the optimal stationary random policy for the expected weighted sum of age of information (EWSAoI) with $q_{ii}$ as the weights (refer to \cite[Section IV C]{kadota_scheduling_2018}).
Let us choose $\pi^{SR}$ in the RHS of inequality (\ref{eqn:mwa_drift}). Note that the optimal stationary random policy schedules sensor $i$ with probability $\sqrt{q_{ii}}/\sum_{j=1}^M \sqrt{q_{jj}}$\cite[Section IV C]{kadota_scheduling_2018}. Hence,
\begin{align}
    \E[\Delta_{\pi^{MWA}}(t)|\b{h}_t] \leq \sum_{i=1}^M\sqrt{q_{ii}} - \frac{\sum_{i=1}^Mq_{ii}(h_i(t-1)+1)}{\sum_{i=1}^M\sqrt{q_{ii}}}.
\end{align}
Using the result from equation \eqref{eqn:pii_ub} in theorem \ref{thm:upper_bound}, we get
\begin{align}
    \E[\Delta_{\pi^{MWA}}(t)|\b{h}_t] \leq \sum_{i=1}^M\sqrt{q_{ii}} - \frac{\sum_{i=1}^M\E[p_{ii}(t)|\b{h}_t]}{\sum_{i=1}^M\sqrt{q_{ii}}}.
\end{align}
We take the expectation on both sides and employ the law of iterated expectation to obtain,
\begin{align}
    \E[\Delta_{\pi^{MWA}}(t)] \leq \sum_{i=1}^M\sqrt{q_{ii}} - \frac{\sum_{i=1}^M\E[p_{ii}(t)]}{\sum_{i=1}^M\sqrt{q_{ii}}}.
\end{align}
From the telescopic nature of $\Delta_{\pi^{MWA}}(t)$, we obtain the following inequality on summing the above expression from $t=1$ to $t=T$.
\begin{align}
    \E[L(T+1)] - \E[L(1)] \leq T\sum_{i=1}^M\sqrt{q_{ii}} - \frac{\sum_{t=1}^{T}\sum_{i=1}^M\E[p_{ii}(t)]}{\sum_{i=1}^M\sqrt{q_{ii}}}.
\end{align}
On rearrangement of terms, we end up with
\begin{align}
    \frac{1}{T}&\sum_{t=1}^{T}\sum_{i=1}^M\E[p_{ii}(t)] \leq \left(\sum_{i=1}^M\sqrt{q_{ii}}\right)^2+\left(\sum_{i=1}^{M}\sqrt{q_{ii}}\right)\frac{\E[L(1)]}{T},
\end{align}
where the last inequality is true because $L(T+1)$ is non-negative. Note that $L(1)$ is a finite quantity, and hence, as $T \to \infty$, we get,
\begin{equation}
    \lim_{T \to \infty}\frac{1}{T}\sum_{t=1}^T\sum_{i=1}^M\E[p_{ii}(t)] \leq \left(\sum_{i=1}^M\sqrt{q_{ii}}\right)^2.
\end{equation}
Therefore,
\begin{align}
    P_{MWA} &= \lim_{T \to \infty}\frac{1}{T}\sum_{t=0}^{T-1}\E\left[||\b{x}_t - \hat{\b{x}}_t||^2_2\right]\nonumber\\
    &=\frac{1}{T}\sum_{t=1}^T\sum_{i=1}^M\E[p_{ii}(t)] - \tr{Q}\nonumber\\
    &\leq \left(\sum_{i=1}^M\sqrt{q_{ii}}\right)^2 - \sum_{i=1}^M q_{ii}.
\end{align}
As discussed before, we use the age-based lower bound for $P_{OPT}$. This is discussed in detail in Section \ref{chap:scaling}. We only state and use the result here.
The optimal estimation error $P_{OPT}$ is lower bounded as
\begin{equation}
    P_{OPT} \geq \frac{1}{2}\left(\sum_{i=1}^M \sqrt{q'_{ii}}\right)^2 - \frac{1}{2}\sum_{i=1}^Mq'_{ii}.
\end{equation}
Therefore, the ratio, $P_{MWA}/P_{OPT}$ can be upper bounded as
\begin{equation}
    \frac{P_{MWA}}{P_{OPT}} \leq 2\frac{\left(\sum_{i=1}^M\sqrt{q_{ii}}\right)^2-\sum_{i=1}^Mq_{ii}}{ \left(\sum_{i=1}^M\sqrt{q'_{ii}}\right)^2 - \sum_{i=1}^M q'_{ii}}.
\end{equation}
\end{proof}

\section{Proof of Theorem \ref{thm:order_of_mag_low}}\label{appendix:thm_order_of_mag_low}
\noindent
To prove theorem \ref{thm:order_of_mag_low}, we will use the following lemmas. 
\begin{lemma}\label{lemma:low_rank_full_rank_lb}
The optimal scheduling problem in (\ref{eqn:kalman_prob_prior}) with non-full rank $Q$ ($\text{rank}(Q) = L < M$) is lower bounded by the relaxed tracking problem
    \begin{equation*}
    \begin{aligned}
    P^* := \min_{\pi \in \Pi} \quad & \lim_{T \to \infty} \frac{1}{T} \sum_{t=1}^T \E[\tr{P'_t}]\\
    \textrm{s.t.}  \quad & P'_{t+1} = P'_t - \frac{P'_t\b{\Tilde{c}}_t\b{\Tilde{c}}_t^TP'_t}{\b{\Tilde{c}}_t^TP'_t\b{\Tilde{c}}_t} + Q'\\
      &\b{\Tilde{c}}_t \in \R^L,
    \end{aligned}
    \end{equation*}
    where $P'_t \succ 0$ and $P'_t \in \R^{L \times L}$, $L < M$ and $Q'$ is full rank in $\R ^{L \times L}$.
\end{lemma}
\begin{proof}
    In Appendix \ref{appendix:low_rank_full_rank_lb}.
\end{proof}
\begin{lemma}\label{lemma:low_rank_full_rank_ub}
There exists a constant $c > 0$ such that
the optimal scheduling problem in (\ref{eqn:kalman_prob_prior}) with non-full rank $Q$ (rank$(Q) = L < M$)
is upper bounded by $c\cdot P^*$, where 
    \begin{equation*}
    \begin{aligned}
    P^* := \min_{\pi \in \Pi} \quad & \lim_{T \to \infty} \frac{1}{T} \sum_{t=1}^T \E[\tr{P'_t}]\\
    \textrm{s.t.}  \quad & P'_{t+1} = P'_t - \frac{P'_t\b{\Tilde{c}}_t\b{\Tilde{c}}_t^TP'_t}{\b{\Tilde{c}}_t^TP'_t\b{\Tilde{c}}_t} + Q'\\
      &\b{\Tilde{c}}_t = \b{e}_j, j \in \{1,...,L\},
    \end{aligned}
    \end{equation*}
    and $P'_t \succ 0$, $P'_t \in \R^{L \times L}$, $L < M$, $Q'$ is a full-rank $L \times L$ submatrix of $Q$.
\end{lemma}
\begin{proof}
    In Appendix \ref{appendix:low_rank_full_rank_ub}.
\end{proof}

\begin{lemma}\label{lemma:full_rank_relaxed_lb}
Consider a relaxation of problem (\ref{eqn:kalman_prob_prior}) where $Q$ is a full rank matrix. If
\begin{equation*}
\begin{aligned}
 P^* := \min_{\pi \in \Pi} \quad & \lim_{T \to \infty} \frac{1}{T} \sum_{t=1}^T \E[tr(P_t)]\\
\textrm{s.t.}  \quad & P_{t+1} = P_t - \frac{P_t\b{c}_t\b{c}_t^TP_t}{\b{c}_t^TP_t\b{c}_t} + Q\\
  &\b{c}_t \in \R^M,
\end{aligned}\label{eqn:kalman_prob_prior_relaxed_1}
\end{equation*}
then,
\begin{equation*}
    P^* \geq \frac{M(M+1)}{2}\lambda_M(Q)
\end{equation*}
where $\lambda_M(Q)$ denotes the smallest eigenvalue of $Q$.
\end{lemma}
\begin{proof}
    In Appendix \ref{appendix:full_rank_relaxed_lb}.
\end{proof}
\noindent
Equipped with the above lemmas, we now proceed to prove theorem \ref{thm:order_of_mag_low}. 
\begin{proof}[Proof of Theorem \ref{thm:order_of_mag_low}]
Using Lemma \ref{lemma:low_rank_full_rank_lb} and equation \eqref{eqn:expec_error_P_t_equivalence}, it is evident that
\begin{align}
    P_{OPT} \geq P'^* - \tr{Q},
\end{align}
where, 
\begin{equation}
\begin{aligned}
 P'^* := \min_{\pi \in \Pi} \quad & \lim_{T \to \infty} \frac{1}{T} \sum_{t=1}^T \E[\tr{P'_t}]\\
\textrm{s.t.}  \quad & P'_{t+1} = P'_t - \frac{P'_t\b{c}_t\b{c}_t^TP'_t}{\b{c}_t^TP'_t\b{c}_t} + Q'\\
  &\b{c}_t \in \R^L,
\end{aligned}\label{eqn:kalman_prob_prior_relaxed_2}
\end{equation}
and $Q'$ is a full-rank submatrix of $Q$. 
Using the result of Lemma \ref{lemma:full_rank_relaxed_lb}, we have 
\begin{align}
    P'^* \geq \frac{L(L+1)}{2}\lambda_{L}(Q').
\end{align}
Since $Q'$ is full rank, $\lambda_L(Q') > 0$.
Using, Lemma \ref{lemma:low_rank_full_rank_ub} and equation \ref{eqn:expec_error_P_t_equivalence}, it is evident that,
\begin{align}
    P_{OPT} \leq c\left(P'^\dag - \sum_{i=1}^LQ'_{ii}\right),
\end{align}
where,
\begin{equation}
\begin{aligned}
 P'^\dag := \min_{\pi \in \Pi} \quad & \lim_{T \to \infty} \frac{1}{T} \sum_{t=1}^T \E[\tr{P'_t}]\\
\textrm{s.t.}  \quad & P'_{t+1} = P'_t - \frac{P'_t\b{c}_t\b{c}_t^TP'_t}{\b{c}_t^TP'_t\b{c}_t} + Q'\\
  &\b{c}_t = \b{e}_j, j \in \{1,...,L\},
\end{aligned}\label{eqn:kalman_prob_prior_relaxed_3}
\end{equation}
and $Q'$ is a full-rank submatrix of $Q$. Using inequality (\ref{eqn:ub_final_full_rank}) in subsection \ref{subsection:full_rank}, we know that $P'^\dag$ is upper bounded as
\begin{equation}
    P'^\dag \leq \left(\sum_{i=1}^L\sqrt{Q'_{ii}}\right)^2
\end{equation}
where $Q'_{ii}$ are the diagonal elements of $Q'$.

\end{proof}

\section{Proof of Lemma \ref{lemma:low_rank_full_rank_lb}}\label{appendix:low_rank_full_rank_lb}
\begin{proof}
Consider the system model (\ref{eqn:sys_model}) $\b{x}_{t+1} = \b{x}_t + \b{w}_t$, where $\b{w}_t \sim \mathcal{N}(0, Q)$, and $\text{rank}(Q) = L < M$. Since there are only $L$ independent rows in $Q$,  we can choose $M-L$ unique row indices and write them as a linear combination of the other rows. Let this index set be $\mathcal{I} = \{i_1,...,i_{M-L}\}$. Now, note that for any $i \in \mathcal{I}$ and for all $k = 1,...,M$, we have linear coeeficients $\alpha^i_j$ such that
    \begin{equation*}
        \E[w^k_tw^i_t] = \E\left[w_t^k\left(\sum_{j\notin \mathcal{I}}\alpha^i_jw^j_t\right)\right].
    \end{equation*}
    Rewriting this, we get for all $k$,
    \begin{equation*}
        \E\left[w^k_t\left(w^i_t - \sum_{j\notin \mathcal{I}}\alpha^i_jw^j_t\right)\right] = 0.
    \end{equation*}
    Since $w^j_t$ are jointly Gaussian, uncorrelatedness implies independence. Hence, $w^k_t \perp \left(w^i_t - \sum_{j\notin \mathcal{I}} \alpha^i_jw^j_t\right)$ for all $k$. But, this is only possible if 
    \begin{equation*}
       w^i_t - \sum_{j\notin \mathcal{I}} \alpha^i_jw^j_t = 0
    \end{equation*} almost surely. This is because the sum $w^i_t - \sum_{j\notin \mathcal{I}} \alpha^i_jw^j_t$ clearly depends on at least one $w^k_t$. Hence, there exists $ w^i_t = \sum_{j\notin \mathcal{I}} \alpha^i_jw^j_t$.
    From the above result, we note that, for $i \notin \mathcal{I}$,
    \begin{equation*}
        x^i_{t+1} - x^i_t = w^i_t
    \end{equation*}
    and, for $i \in \mathcal{I}$,
    \begin{align*}
        x^i_{t+1} - x^i_t  &= w^i_t = \sum_{j\notin \mathcal{I}}\alpha^i_jw^j_t = \sum_{j\notin \mathcal{I}}\alpha^i_j(x^j_{t+1} - x^j_t).
    \end{align*}
    Since $x_0 = 0$ without loss in generality, by summing the above telescopic sequence from 0 to $t-1$, we get,
    \begin{equation*}
        x^i_t = \sum_{j\notin \mathcal{I}}\alpha^i_jx^j_t
    \end{equation*}
    almost surely. Assume without loss in generality $i_1=L+1, i_2 = L+2,..., i_{M-L} = M$. For notational simplicity, define matrix $\alpha$ as
    \begin{equation}
        \alpha := \begin{bmatrix}
        \alpha^{i_1}_1 & \alpha^{i_1}_2 & \hdots & \alpha^{i_1}_L\\
        \alpha^{i_2}_1 & \alpha^{i_2}_2 & \hdots & \alpha^{i_2}_L\\
        \vdots & \vdots&\ddots & \vdots\\
        \alpha^{i_{M-L}}_1 & \alpha^{i_{M-L}}_2 & \hdots & \alpha^{i_{M-L}}_L
        \end{bmatrix}
    \end{equation}
    Consequently, we define a reduced state $\b{x}'_t$ with $L$ dimensions, such that $x'^j_t = x^j_t$ for $j \in \{1,...,L\}$, and a reduced state innovation, $v_t$, where $v^j_t = w^j_t$ for $j \in \{1, ..., L\}$. Rewriting, 
    \begin{align}\label{eqn:mat_mult_for_state_reduc}
        x_t &= \begin{bmatrix}
        I\\
        \alpha
        \end{bmatrix} \b{x}'_t = H\b{x}'_t\\
        w_t &= \begin{bmatrix}
        I\\
        \alpha
        \end{bmatrix} v_t = Hv_t.
    \end{align}
    The system dynamics can now be represented with the reduced order system 
    \begin{equation}\label{eqn:reduced_state_eqn}
        \b{x}'_{t+1} = \b{x}'_{t} + \b{v}_t.
    \end{equation}
    This is because we can multiply (\ref{eqn:reduced_state_eqn}) with $A$ to get the original system. 
    \begin{equation*}
    \b{x}_{t+1} = H\b{x}'_{t+1} = H\b{x}'_t + H\b{v}_t = \b{x}_t + \b{w}_t    
    \end{equation*}
    Also, note that $P_t = \E[(\b{x}_t-\b{\hat{x}}_t)(\b{x}_t -\b{\hat{x}}_t)^T]$. 
    Consequently, we define $P'_t := \E[(\b{x}'_t-\hat{\b{x}}'_t )(\b{x}'_t -\hat{\b{x}}'_t)^T]$, 
    and $Q'= \E[\b{v}_t\b{v}_t^T]$. Note that $Q = \E[\b{w}_t\b{w}_t^T] = \E[(H\b{v}_t)(\b{v}_t^TH^T)] = HQ'H^T$.\\\\
\noindent
Since the first $L$ dimensions of $x_t$ and $\Tilde{x}_t$ are the same, 
\begin{align}
    \tr{P_t} = \sum_{i=1}^M \E[(x^i_t)^2] &\geq \sum_{i=1}^L \E[(x^i_t)^2] = \sum_{i=1}^L \E[(x'^i_t)^2] = \tr{P'_t}.\label{eqn:trace_inequality}
\end{align}
In the original scheduling problem, note that scheduling $x^i_t$ when $i \notin \mathcal{I}$, is the same as scheduling $\Tilde{x}^i_t$ or in other words, the scheduling vector $\b{\Tilde{c}}_t = \b{e}_i$. But scheduling $x^i_t$ when $i \in \mathcal{I}$, is equivalent to scheduling with $\b{\Tilde{c}}_t^T = [\alpha^i_1\;...\;\alpha^i_L ]^T$ in the reduced state problem. Using this equivalence and inequality (\ref{eqn:trace_inequality}), we can lower bound the original problem by using the reduced state but relaxing the scheduling constraint to any linear combination.
\end{proof}

\section{Proof of Lemma \ref{lemma:low_rank_full_rank_ub}}\label{appendix:low_rank_full_rank_ub}
\begin{proof}
As seen in the proof of Lemma \ref{lemma:low_rank_full_rank_lb} in Appendix \ref{appendix:low_rank_full_rank_lb}, we observe that the $M$-dimensional system model in (\ref{eqn:sys_model}), $\b{x}_{t+1} = \b{x}_t + \b{w}_t$, with $\b{w}_t \sim \mathcal{N}(0, Q)$, $\text{rank}(Q) = L < M$, can be converted to a $L$-dimensional reduced state system model (\ref{eqn:reduced_state_eqn}), $\b{x}'_{t+1} = \b{x}'_t + \b{v}_t$, where $\b{v}_t \sim \mathcal{N}(0, Q')$. Note that on multiplying matrix $A$ (as described in \ref{eqn:mat_mult_for_state_reduc}) with $\b{x}'_t$ and $\b{v}_t$, we get $\b{x}_t$ and $\b{w}_t$ respectively. Define $P_t := \E[(\b{x}_t-\b{\hat{x}}_t)(\b{x}_t-\b{\hat{x}}_t)^T]$ and $P'_t := \E[(\b{x}'_t-\hat{\b{x}}'_t)(\b{x}'_t-\hat{\b{x}}'_t)^T]$. \\

\noindent
Assume that at time $t$, $P_t = AP'_tA^T$. Then at time $t+1$,
\begin{align}
    P_{t+1} &= P_t - \frac{P_t\b{c}_t\b{c}_t^TP_t}{\b{c}_t^TP_t\b{c}_t} + Q\nonumber\\
    &= HP'_tH^T - \frac{HP'_tH^T\b{c}_t\b{c}_t^THP'_tH^T}{\b{c}_t^TAP'_tA^T\b{c}_t} + A\Tilde{Q}A^T\nonumber \\
    &= H\left(P'_t - \frac{P'_tH^T\b{c}_t\b{c}_t^THP'_t}{\b{c}_t^THP'_tH^T\b{c}_t} + \Tilde{Q}\right)H^T
\end{align}
If we do not schedule sensors $i \in \mathcal{I}$ (as defined in Appendix \ref{appendix:low_rank_full_rank_lb}), i.e., for $i \in \mathcal{I}$, $c^i_t = 0$, then, note that
\begin{align}
    H^T\b{c}_t = \begin{bmatrix}
        I &
        \alpha^T 
    \end{bmatrix} \b{c}_t = \begin{bmatrix}
        I &
        \alpha^T 
    \end{bmatrix}\begin{bmatrix}
        \b{\Tilde{c}}_t \\
        \b{0}
    \end{bmatrix} = \b{\Tilde{c}}_t.
\end{align}
Consequently,
\begin{align}
    P_{t+1} &= H\left(P'_t - \frac{P'_t\b{\Tilde{c}}_t\b{\Tilde{c}}_t^THP'_t}{\b{c}_t^TP'_tH^T\b{\Tilde{c}}_t} + \Tilde{Q}\right)H^T = HP'_{t+1}H^T.
\end{align}
If we assume $P_0 = HP'_0H^T$, then, if we do not schedule $i \in \mathcal{I}$, we have $P_t = HP'_tH^T$, $\forall t$.
\noindent
Now, we observe that, if $i \in \mathcal{I}$ is not scheduled, then
\begin{align}
    \tr{P_t} &= \tr{HP'_tH^T} =\tr{H^THP'_t}\nonumber\\ 
    &\leq \lambda_1(H^TH)\tr{P'_t} \leq\lambda_1(I + \alpha^T\alpha)\tr{P'_t}.\nonumber
    \end{align}
Hence, $\tr{P_t} \leq c\cdot \tr{P'_t}$.
Consequently, 
\begin{align}
    \lim_{T \to \infty}\frac{1}{T}\sum_{t = 1}^T\E[\tr{P_t}] &\leq c\cdot \lim_{T \to \infty}\frac{1}{T}\sum_{t = 1}^T\E[\tr{P'_t}].
\end{align}
Observe that the inequality is true for any policy that does not schedule $i \in \mathcal{I}$. So,
\begin{align}
    \lim_{T \to \infty}\frac{1}{T}&\sum_{t = 1}^T\E_\pi[\tr{P_t}] \leq c\cdot \min_{\pi \in \Pi} \lim_{T \to \infty}\frac{1}{T}\sum_{t = 1}^T\E_\pi[\tr{P'_t}] = c P^*.
\end{align}
But any policy that does not schedule $i \in \mathcal{I}$ is not going to be any better than the optimal policy. Therefore, we have 
\begin{align}
    P_{OPT} \leq \lim_{T \to \infty}\frac{1}{T}\sum_{t = 1}^T\E_\pi[\tr{P_t}],
\end{align}
and hence, $P_{OPT} \leq  c P^*$.
\end{proof}

\section{Proof of Lemma \ref{lemma:full_rank_relaxed_lb}}\label{appendix:full_rank_relaxed_lb}
\begin{proof}
From the recursion 
    \begin{equation}
        P_{t+1} = P_t + Q - \frac{P_t\b{c}_t\b{c}_t^TP_t}{\b{c}_t^TP_t\b{c}_t}\label{eqn:recursion},
    \end{equation}
    we observe that $P_{t+1}$ is a rank one ``perturbation" of $P_t + Q$. 
    Let $\lambda_1(A) \geq \lambda_2(A) \geq ... \geq \lambda_M(A)$ denote the ordered eigenvalues of a symmetric matrix $A$. This is possible if $A$ is a symmetric matrix as all the eigenvalues are real. We now employ Weyl's inequalities \cite{horn_johnson_1985} to bound the the ordered eigenvalues of the symmetric matrices $P_{t}$. 
    \begin{lemma} \label{lemma:lb_induction}
        If at time $t=0$, $P_0 \succeq 0$, then $\lambda_i(P_{t}) \geq  (M-i+1)\lambda_{M}(Q),\;  i  = 1,...,M$ after finite time.
    \end{lemma}
    \begin{proof}
        In Appendix \ref{appendix:lb_induction}.
    \end{proof}
    From Lemma \ref{lemma:lb_induction}, it is evident that $\tr{P_t} = \sum_{i=1}^M\lambda_i(P_t) \geq \sum_{i=1}^M (M-i+1)\lambda_M(Q) = M(M-1)\lambda_M(Q)/2$ after some finite time $t_0$. Hence, for any policy, we have
    \begin{align}
        \lim_{T \to \infty }\frac{1}{T}\sum_{t=1}^T \tr{P_t} &= \lim_{T \to \infty }\frac{1}{T}\sum_{t=1}^{t_0-1} \tr{P_t} \nonumber\\
        & \;\;\;\;\;\;\;\;\;+  \lim_{T \to \infty }\frac{T-t_0}{T}\frac{1}{T-t_0}\sum_{t=t_0}^T \tr{P_t}\nonumber \\
        &\geq \frac{M(M-1)}{2}\lambda_M(Q)
    \end{align}
    The first term in the RHS goes to zero as $T \to \infty$ as the sum is finite. For $t \geq t_0$, $\tr{P_t} \geq M(M-1)\lambda_M(Q)/2$. So,
    \begin{align}
         \lim_{T \to \infty }\frac{1}{T}\sum_{t=1}^T \tr{P_t} \geq \frac{M(M-1)}{2}\lambda_M(Q)
    \end{align}
  Since this is true for any policy, we conclude
    \begin{align}
        \min_{\pi \in \Pi}\lim_{T\to \infty}\frac{1}{T}\sum_{t=1}^T 
        \E_\pi[\tr{P_t}] \geq \frac{M(M-1)}{2}\lambda_M(Q).
    \end{align}
\end{proof}

\section{Proof of Lemma \ref{lemma:lb_induction}}\label{appendix:lb_induction}
\begin{proof}
    According to Weyl's inequality\cite{horn_johnson_1985}, for any two symmetric matrices $A,B \in \R^{M\times M}$,
    \begin{align}
        \lambda_{i}(A) + \lambda_{j}(B) &\leq \lambda_{i+j-M}(A + B)\label{eqn:weyl_lb}\\
        \lambda_{i+j-1}(A + B) &\leq  \lambda_{i}(A) + \lambda_{j}(B)\label{eqn:weyl_ub}
    \end{align}
    where $1 \leq i, j \leq M$ and $i+j \leq M + 1$.\\
    Choosing $j$ to be $1$ and using Weyl's on (\ref{eqn:recursion}), we get
    \begin{equation*}
        \lambda_{i}(P_{t+1}) \leq \lambda_i(P_t + Q) + \lambda_1\left(-\frac{P_t\b{c}_t\b{c}_t^TP_t}{\b{c}_t^TP_t\b{c}_t}\right) = \lambda_i(P_t + Q) 
    \end{equation*}
    because the largest eigenvalue of $-\frac{P_t\b{c}_t\b{c}_t^TP_t}{\b{c}_t^TP_t\b{c}_t}$ is 0. Similarly, choosing $j$ to be $M-1$ and using Weyl's on (\ref{eqn:recursion}), we get,
     \begin{equation*}
        \lambda_{i-1}(P_{t+1}) \geq \lambda_i(P_t + Q) + \lambda_{M-1}\left(-\frac{P_t\b{c}_t\b{c}_t^TP_t}{\b{c}_t^TP_t\b{c}_t}\right) = \lambda_i(P_t + Q) 
    \end{equation*}
    because the second smallest eigenvalue of $-\frac{P_t\b{c}_t\b{c}_t^TP_t}{\b{c}_t^TP_t\b{c}_t}$ is 0, as it is a rank 1 negative semidefinite matrix. On combining both the inequalities, we have,
    \begin{equation}
        \lambda_{i+1}(P_t + Q) \leq \lambda_i(P_{t+1}) \leq \lambda_{i}(P_t + Q) \label{eqn:interlacing}
    \end{equation}
    Using (\ref{eqn:interlacing}) and lemma (\ref{lemma:pt_geq_q}), we can derive the required bound using induction.\\
    At time $t=0$, $P_0 \succ 0$. Therefore, at $t = 1$, $P_1 = P_0 - P_0\b{c}_t\b{c}_t^TP_0/\b{c}_t^TP_0\b{c}_t + Q \succeq Q$ (Lemma \ref{lemma:pt_geq_q}). Hence, $\lambda_i(P_1) \geq \lambda_i(Q) \;\; \forall i = 1,...,M$, and $\lambda_i(P_1) \geq \lambda_M(Q)$.
    At time $t = 2$,
    \begin{itemize}
        \item $\lambda_M(P_2) \geq \lambda_M(Q)$
        \item For $i =1,...,M-1$, $\lambda_i(P_2) \geq \lambda_{i+1}(P_1) + \lambda_M(Q) \geq 2\lambda_M(Q)$. 
    \end{itemize}
    At time $t = 3$, 
    \begin{itemize}
        \item $\lambda_M(P_3) \geq \lambda_M(Q)$
        \item $\lambda_{M-1}(P_3) \geq \lambda_M(P_2) + \lambda_M(Q) \geq 2\lambda_M(Q)$ 
        \item For $i =1,...,M-2$, $\lambda_i(P_3) \geq \lambda_{i+1}(P_2) + \lambda_M(Q) \geq 3\lambda_M(Q)$. 
    \end{itemize}
    On extending the above line of argument inductively till time $t = M$, we get $\lambda_i(P_M) \geq (M-i+1)\lambda_M(Q)$.\\
    Now, assume the hypothesis that at some $t$, $\lambda_i(P_t) \geq (M-i+1)\lambda_M(Q)$. This is true for $t = M$ as argued above. Note that at time $t + 1$,
    \begin{align}
        \lambda_i(P_{t+1}) &\geq  \lambda_{i+1}(P_{t} + Q) \nonumber\\
        &\stackrel{(a)}{\geq} \lambda_{i+1}(P_{t}) + \lambda_M(Q)\nonumber\\
        &\geq (M-(i+1)+1)\lambda_M(Q) + \lambda_M(Q)\nonumber\\
        &= (M-i+1)\lambda_M(Q)
    \end{align}
    where inequality $(a)$ comes from employing the Weyl's inequality on $\lambda_{i+1}(P_{t} + Q)$.
    Also note that $\lambda_M(P_{t + 1}) \geq \lambda_M(Q)$. Therefore, the hypothesis is true for one step. Note that $\lambda_M(P_{t+1}) \geq \lambda_M(Q)$ is always true from Lemma \ref{lemma:pt_geq_q}. Therefore, we have proved the lemma by induction.
\end{proof}

\section{Proof of Lemma \ref{lemma:eig_p_t_gen_A}}\label{appendix:eigenvalue_inequality}
\begin{proof}
    The proof of this lemma is very similar to \ref{appendix:lb_induction}. Note that with a general $A$ matrix, the dynamics of $P_t$ follows
    \begin{align}
        P_{t+1} = AP_{t}A^T - \frac{AP_t\b{c}_t\b{c}_t^TP_tA^T}{\b{c}_t^TP_t\b{c}_t} + Q.
    \end{align}
    Again, observe that $P_{t+1}$ is a rank one ``perturbation" of $AP_tA^T +Q$. The next step involves using Weyl's inequalities (as described in \ref{appendix:lb_induction}). 
    \begin{equation}
        \lambda_{i+1}(AP_tA^T + Q) \leq \lambda_i(P_{t+1}) \leq \lambda_{i}(AP_tA^T + Q) \label{eqn:interlacing1}
    \end{equation}
    Using (\ref{eqn:interlacing1}), we can derive the required bound using induction.\\
    At time $t=0$, $P_0 \succ 0$. Therefore, at $t = 1$, $P_1 = AP_0A^T - AP_0\b{c}_t\b{c}_t^TP_0A^T/\b{c}_t^TP_0\b{c}_t + Q \succeq Q$ (this is a straightforward extension of Lemma \ref{lemma:pt_geq_q}). Hence, $\lambda_i(P_1) \geq \lambda_i(Q) \;\; \forall i = 1,...,M$, and $\lambda_i(P_1) \geq \lambda_M(Q)$.
    Before we proceed, we need invoke results from \cite[Theorem 5.4]{ipsen_1998} on the relative perturbations of eigenvalues on symmetric transformations. For any matrices $X$ and $D$,
    \begin{align}\label{eqn:mult_perturb}
        \lambda_{\min}(DD^T)&\lambda_j(X) \nonumber\\
        &\leq \lambda_{j}(DXD^T) \nonumber\\
        &\leq \lambda_{\max}(DD^T)\lambda_j(X) &&\text{(From \cite[Theorem 5.4]{ipsen_1998})}.
    \end{align} 
    Using the above result, at time $t = 2$,
    \begin{itemize}
        \item $\lambda_M(P_2) \geq \lambda_M(Q)$
        \item For $i =1,...,M-1$, 
        \begin{align}
            \lambda_i(P_2) &\geq \lambda_{i+1}(AP_1A^T) + \lambda_M(Q)\nonumber\\
            &\geq \lambda_M(AA^T) \lambda_{i+1}(P_1) + \lambda_M(Q)\nonumber\\ 
            &\geq (\lambda_M(AA^T) + 1)\lambda_M(Q) =\frac{\lambda_{M}(AA^T)^2 - 1}{\lambda_{M}(AA^T) - 1}\lambda_M(Q). 
        \end{align} 
    \end{itemize}
    At time $t = 3$, 
    \begin{itemize}
        \item $\lambda_M(P_3) \geq \lambda_M(Q)$
        \item $\lambda_{M-1}(P_3) \geq \lambda_M(AP_2A^T) + \lambda_M(Q) \geq \lambda_M(AA^T)\lambda_{M}(P_2) + \lambda_M(Q)\geq \frac{\lambda_{M}(AA^T)^2 - 1}{\lambda_{M}(AA^T) - 1}\lambda_M(Q)$. 
        \item For $i =1,...,M-2$, $\lambda_i(P_3) \geq \lambda_{i+1}(AP_2A^T) + \lambda_M(Q) \geq \lambda_M(AA^T)\lambda_{i+1}(P_2) + \lambda_M(Q)\geq \frac{\lambda_{M}(AA^T)^3 - 1}{\lambda_{M}(AA^T) - 1}\lambda_M(Q)$. 
    \end{itemize}
    On extending the above line of argument inductively till time $t = M$, we get
    \begin{equation}
        \lambda_i(P_M) \geq \frac{\lambda_M(AA^T)^{M-i+1} - 1}{\lambda_M(AA^T) - 1}\lambda_M(Q).
    \end{equation}
    Now, assume the hypothesis that at some $t$, $\lambda_i(P_t) \geq \frac{\lambda_M(AA^T)^{M-i+1} - 1}{\lambda_M(AA^T) - 1}\lambda_M(Q)$. This is true for $t = M$ as argued above. Note that at time $t + 1$,
    \begin{align}
        \lambda_i(P_{t+1}) &\geq  \lambda_{i+1}(AP_{t}A^T + Q) \nonumber\\
        &\stackrel{(a)}{\geq} \lambda_M(AA^T)\lambda_{i+1}(P_{t}) + \lambda_M(Q)\nonumber\\
        &\geq \frac{\lambda_M(AA^T)^{M-i} - 1}{\lambda_M(AA^T) - 1}\lambda_M(Q) + \lambda_M(Q)\nonumber\\
        &= \frac{\lambda_M(AA^T)^{M-i+1} - 1}{\lambda_M(AA^T) - 1}\lambda_M(Q)
    \end{align}
    where inequality $(a)$ comes from employing the Weyl's inequality and \eqref{eqn:mult_perturb} on $\lambda_{i+1}(AP_{t}A^T + Q)$.
    Also note that $\lambda_M(P_{t + 1}) \geq \lambda_M(Q)$. Therefore, the hypothesis is true for one step. Note that $\lambda_M(P_{t+1}) \geq \lambda_M(Q)$ is always true from the straightforward extension of Lemma \ref{lemma:pt_geq_q}. Therefore, we have proved the lemma by induction.
\end{proof}

\section{Inverse Inequality of Positive Semidefinte Matrices}\label{appendix:inv_psd}
\begin{lemma}\label{lemma:inv_psd}
If $A$ and $B$ are invertible positive definite matrices such that $A - B \succeq 0$, then $B^{-1} - A^{-1} \succeq 0$.
\end{lemma}
\begin{proof}
    Since $A$ and $B$ are positive definite matrix, we can factorize $A$ as $A=(A^{1/2})^2$ and $B$ as $B = (B^{1/2})^2$, and $B^{1/2} \succ 0$. This implies $B^{-1/2} := (B^{1/2})^{-1}$ exists and $B^{-1/2} \succ 0$.
    Since $A-B \succeq 0$, on pre- and post- multiplying by a positive definite matrix $B^{-1/2}$, we still preserve positive definiteness. Hence,
    \begin{equation}
        B^{-1/2}AB^{-1/2} - I \succeq 0.
    \end{equation}
    Note that the above matrix inequality implies that all the eigenvalues of $ B^{-1/2}AB^{-1/2}$ are greater than 1.
    Note that $B^{-1/2}AB^{-1/2} = (B^{-1/2}A^{1/2}) (A^{1/2}B^{-1/2})$. 
    Since the eigenvalues of a matrix product are invariant to switching the order of multiplication \cite{strang_eigenvalues_nodate, dobrushkin_eigenvalues_nodate}, $(A^{1/2}B^{-1/2})(B^{-1/2}A^{1/2}) = A^{1/2}B^{-1}A^{1/2}$ has the same eigenvalues as $B^{-1/2}AB^{-1/2}$. Therefore, all eigenvalues of $A^{1/2}B^{-1}A^{1/2}$ are greater than 1, which in turn suggests,
    \begin{align}
        A^{1/2}B^{-1}A^{1/2} - I \succeq 0.
    \end{align}
    On pre- and post- multiplying by $A^{-1/2} \succ 0$, we obtain,
    \begin{align}
        B^{-1} - A^{-1} \succeq 0.
    \end{align}
\end{proof}
\end{document}